\documentclass[journal,11pt,twocolumn,romanappendices]{IEEEtran}

\ifCLASSINFOpdf
   \usepackage[pdftex]{graphicx}

\else

\fi

\usepackage{amssymb,bbm}

\usepackage{amsmath}

\usepackage{color}

\def \ver{techrep}
\def \verreview{review}
\def \vertechrep{techrep}

\newcommand{\changes}[1]{%
	\ifx\ver\verreview%
		{\textcolor{blue}{#1}}%
	\else {#1}%
	\fi%
}

\newcommand{\citeapp}[1]{%
	\ifx\ver\vertechrep%
		{Appendix \ref{#1}}%
	\else {\cite[Appendix \ref{#1}]{implicit_arxiv}}%
	\fi%
}

\ifx\ver\vertechrep%
	\def\myFigureScale{0.8}%
\else \def\myFigureScale{0.6}%
\fi%

\newcommand{\Pa}{\mathcal P}
\newcommand{\x}{\mathbf x}

\newcommand{\X}{\mathbf X}
\newcommand{\y}{\mathbf y}
\newcommand{\z}{\mathbf z}
\newcommand{\vgamma}{\boldsymbol \gamma}

\newcommand{\E}{{\mathbb E}}

\newcommand{\lru}{\textproc{LRU}}
\newcommand{\lfu}{\textproc{LFU}}
\newcommand{\lruone}{\textproc{LRU-One}}

\newcommand{\lruall}{\textproc{LRU-All}}

\newcommand{\qlru}{\textproc{$q$LRU}}
\newcommand{\qlrulazy}{\textproc{$q$LRU-Lazy}}
\newcommand{\fifo}{\textproc{FIFO}}
\newcommand{\random}{\textproc{RANDOM}}
\newcommand{\twolru}{\textproc{2LRU}}

\newcommand{\qlrud}{\mbox{\textproc{$q$LRU-$\Delta$}}}
\newcommand{\qlrudd}{\textproc{$q$LRU-$\Delta d$}}
\newcommand{\qlrudh}{\textproc{$q$LRU-$\Delta h$}}
\newcommand{\greedyd}{\textproc{Greedy-$d$}}
\newcommand{\greedyh}{\textproc{Greedy-$h$}}

\usepackage{amsthm}

\theoremstyle{plain}
\newtheorem{thm}{Theorem}[section]
\newtheorem{lem}[thm]{Lemma}
\newtheorem{prop}[thm]{Proposition}
\newtheorem{cor}{Corollary}

\makeatletter
\newtheorem*{rep@prop}{\rep@title}
\newcommand{\newrepprop}[2]{%
\newenvironment{rep#1}[1]{%
 \def\rep@title{#2 \ref{##1}}%
 \begin{rep@prop}}%
 {\end{rep@prop}}}
\makeatother

\newrepprop{prop}{Proposition}

\theoremstyle{definition}
\newtheorem{defn}{Definition}[section]

\theoremstyle{remark}

\usepackage{algpseudocode,algorithm,algorithmicx}

\usepackage{url}
\usepackage{hyperref}

\ifCLASSOPTIONcompsoc
  \usepackage[caption=false,font=normalsize,labelfont=sf,textfont=sf]{subfig}
\else
  \usepackage[caption=false,font=footnotesize]{subfig}
\fi

\DeclareMathOperator*{\maxim}{maximize}

\usepackage{epstopdf}
\usepackage{verbatim}
\usepackage{algpseudocode}

\begin{document}

\title{A Swiss Army Knife for Caching \\in Small Cell Networks}

\author{
\IEEEauthorblockN{Giovanni Neglia\IEEEauthorrefmark{1}, Emilio Leonardi\IEEEauthorrefmark{3}, Guilherme Ricardo\IEEEauthorrefmark{2}\IEEEauthorrefmark{1}, and Thrasyvoulos Spyropoulos\IEEEauthorrefmark{2}\\
}
\IEEEauthorblockA{\IEEEauthorrefmark{1}Inria, Universit\'e C\^ote d'Azur, France, giovanni.neglia@inria.fr}\\
\IEEEauthorblockA{\IEEEauthorrefmark{3}Politecnico di Torino, Italy, leonardi@polito.it}\\
\IEEEauthorblockA{\IEEEauthorrefmark{2}Eurecom, France, firstname.lastname@eurecom.fr}
}

\maketitle

\begin{abstract}
We consider a dense cellular network, in which a limited-size cache is available at every base station (BS). Coordinating content allocation across the different caches can lead to significant performance gains, but is a difficult problem even when full information about the network and the request process is available. In this paper we present \qlrud, a general-purpose online caching policy that can be tailored to optimize different performance metrics also in presence of coordinated multipoint transmission techniques. The policy requires neither direct communication among BSs, nor a priori knowledge of content popularity and, under stationary request processes, has provable performance guarantees.
\end{abstract}

\section{Introduction}
\label{s:intro}
In the last years, we have witnessed  a dramatic  shift of traffic at  network edge, from  the  wired/fixed component to the wireless/mobile segment.
 This trend, mainly due  to the huge success of mobile devices (smartphones, tablets) and their pervasive applications  (Whatsapp, Instagram, Netflix, Spotify, Youtube, etc.), 
is expected to  further strengthen  in the next  few  years, as testified by several    traffic forecasts. For example 
according to CISCO  \cite{CISCO} in the 5 year interval   ranging from 2017 to 2022 traffic demand on the cellular network 
will approximately  increase by a factor of 9. 
As a consequence, the access (wireless and wired) infrastructure must be completely redesigned by densifying the cellular structure, and moving content closer to users. To this end,  the  massive deployment of  caches within   base stations  of the cellular 
network   is essential to effectively reduce the load on   the back-haul links, as well as limit latency perceived by the user. 

This work  considers a dense cellular network scenario, where caches are placed at every Base Station (BS) and a significant fraction of users  is ``covered'' by  several BSs (whose cells are said to ``overlap'').  The BSs in the transmission range of a given user can coordinate to offer a seamless optimized caching service to the user  and possibly exploit coordinated multipoint (CoMP) techniques~\cite{lee12} on the radio access.
We remark that, as soon as there are overlapping BSs, finding the optimal {offline} static content allocation becomes an NP-hard problem, even when the request process is known, the metric to optimize is the simple cache hit ratio, and coordinated transmissions are not supported~\cite{shanmugam13}.
But realistic scenarios are more complex: popularities are dynamic and unknown a priori and more sophisticated metrics (e.g., PHY-based ones) that further couple nearby BSs-caches are of interest. Moreover, centralized coordination of hundreds or thousands of caches per km$^2$ (e.g., in ultra-dense networks) is often infeasible or leads to excessive coordination overhead.

In  such a  context,  our paper provides an answer  to the open question about the existence  of general  (computationally efficient) distributed strategies for   edge-cache coordination,  which are able to provide some guarantees on global performance  metrics (like  hit ratio,  retrieval time,  load on the servers, etc.). 
In particular, we propose  a new policy---\qlrud---which provably achieves  a locally optimal configuration for general   performance metrics. 

\qlrud{} requires a simple modification to the basic behaviour of \qlru~\cite{garetto16}. 
Upon a hit at a cache, \qlrud{} moves the corresponding content to the front of the queue with a probability that is proportional to the marginal utility of storing this copy. Upon a miss, it introduces the new content with some probability $q$. \qlrud{} inherits from \qlru{} $\mathcal O(1)$ computation time per request and memory requirements proportional to the cache size. Its request-driven operation does not need a priori knowledge of content popularities, removing a limit of most previous work. 
Some information about the local neighborhood (e.g., how many additional copies of the content are stored at close-by caches also serving that user) may be needed to compute the marginal gain. Such information, however, is limited, and can be piggybacked on existing messages the user sends to query such caches, or even on channel estimates messages mobile devices regularly send to nearby BSs~\cite{LTE-book}. 
As an example, we show that \qlrud{} is a practical solution to optimize hit ratio, retrieval time, load on the servers, etc., both when a single BS satisfies the user's request and when multiple BSs coordinate their transmissions through CoMP techniques.


\subsection{Related work}
We limit ourselves to describe work that  specifically addresses the caching problem  in dense cellular networks.

The idea of coordinating the  placement of contents at caches, which are closely located at BSs, was first proposed in~\cite{Caire12} and its extension \cite{shanmugam13} under the name of FemtoCaching. This work assumes that requests follow the Independent Reference Model (IRM) and  geographical popularity profiles are available, i.e., content requests are independent and request rates are known for all cell areas and their intersections. Finding the optimal content placement that maximizes the hit ratio is proved to be 
an NP-hard problem, but a greedy heuristic algorithm is shown to guarantee a $\frac{1}{2}$-approximation of the maximum hit ratio. 
In \cite{Poularakis14},  the authors generalized the approach of \cite{Caire12,shanmugam13}, providing a formulation for the joint content-placement and user-association problem that   maximizes the hit ratio. They also proposed efficient heuristic solutions. {This line of work has been further extended in~\cite{saputra19}, which also considers the request routing problem.}
Authors of \cite{Naveen15}  included the bandwidth costs  in the formulation,
and  proposed an on-line algorithm for the solution of the resulting problem.
In \cite{Chattopadhyay18}, instead, the authors designed a  distributed 
algorithm based on Gibbs sampling, which was shown to asymptotically converge to the optimal allocation.
Reference~\cite{Anastasios2}~revisits the  optimal  content placement problem within a 
stochastic geometry  framework and 
derives an elegant analytical characterization of the optimal policy and its performance. In \cite{avrachenkov17} the authors developed a few asynchronous 
distributed content placement algorithms with polynomial complexity
and limited  communication overhead (communication takes place only between overlapping cells), whose performance was shown to be very good 
in most of the tested scenarios. Still,  they assumed  that  content popularities are perfectly known by the system. Moreover they focused   on cache hit rates, and did not consider CoMP.

One of the first papers that jointly considers caching and CoMP techniques was~\cite{ao15}: two BSs storing the same file can coordinate its transmission to the mobile user in order to reduce the delay or to increase the throughput. The authors considered two caching heuristics: a randomized caching
policy combined with maximum ratio transmission precoding and a threshold policy combined with zero forcing beamforming. 
These policies are in general suboptimal with no theoretical performance guarantee. Reference~\cite{tuholukova17} addresses this issue for joint transmissions techniques. The authors proved that delay minimization leads to a submodular maximization problem as long as the backhaul delay is larger than the transmission delay over the wireless channel. Under such condition, the greedy algorithm provides again a guaranteed approximation ratio. Reference~\cite{chen17} considers two different CoMP techniques, i.e.,~joint transmission and parallel transmission, and derives formulas for the hit rate using tools from stochastic geometry.

Nevertheless, all aforementioned works hold the limiting assumption in~\cite{Caire12} that geographical content popularity profiles are  known by the system. 
Reliable popularity estimates over small geographical areas may be very hard to obtain~\cite{leconte16}.
On the contrary, policies like \lru{} and its variants (\qlru, \textsc{2LRU}, \dots) do not rely on popularity estimation and are known to 
well behave under time-varying popularities. For this reason they are a de-facto standard in most of the deployed caching systems.
Reference~\cite{giovanidis16} proposes a generalization of \lru{} to a dense cellular scenario. As above, a user at the intersection of multiple cells can check the availability 
of the content at every covering cell and then download from one of them. 
The difference with respect  to standard \lru{} is how cache states are updated. In particular, the authors of~\cite{giovanidis16} considered two schemes: \lruone{} and \lruall. In \lruone, each user is assigned to a reference cell/cache and only the state of her reference cache is updated upon a hit or a miss, independently from which cache the content has been retrieved from. In \lruall, the state of all caches covering the user is updated. 

Recently, \cite{paschos19}  proposed a novel approach to design coordinated caching polices in the framework of online linear optimization. A projected gradient method is used to tune the fraction of each content to be stored in a cache and regret guarantees are proved. Unfortunately, this solution requires to store pseudo-random linear combinations of original file chunks, and, even ignoring the additional cost of coding/decoding, it has $\mathcal O(F)$ computation time per request as well as $\mathcal O(F)$ memory requirements, where $F$ is the catalogue size. Also, coding excludes the possibility to exploit CoMP techniques, because all chunks are different. {A caching algorithm  resorting on a  deep  reinforcement  learning approach was instead recently proposed in \cite{wu2019dynamic}}.

Lastly, reference \cite{leonardi18jsac} proposes a novel approximate analytical approach  to study systems of interacting caches, under different caching policies, whose predictions are surprisingly accurate. The framework builds upon the well known  characteristic time  approximation~\cite{che02} for individual caches as well as an exponentialization approximation. We also rely on the same approximations, which are described in Sect.~\ref{s:optimality}. \cite{leonardi18jsac}~also proposes the policy \qlrulazy{}, whose adoption in a dense cellular scenario is shown to achieve hit ratios very close to those offered by the greedy scheme proposed in~\cite{Caire12} even without information about popularity profiles.
\qlrud{} generalizes  \qlrulazy{} to different metrics as well as CoMP transmissions.
{Furthermore, the  analytical results about optimality obtained in this paper, in which we adopt a different technique, are  significantly stronger, while  more elegant and concise.
In this paper, indeed, we prove  global optimality  for \qlrud{}  while in   \cite{leonardi18jsac} only local optimality has been shown for \qlrulazy{}.}

\subsection{Paper Contribution}
\label{s:contri}
The main contribution of this paper is the proposal of \qlrud, a general-purpose caching policy that can be tailored to optimize different performance metrics. The policy implicitly coordinates caching decisions across different caches also taking into account joint transmission opportunities. \qlrud{} is presented in details in Sect.~\ref{s:operation}, after the introduction of  our network model in~Sect.~\ref{s:network_model}.

Sect.~\ref{s:optimality} is devoted to prove that, under a stationary request process, 
\qlrud{} achieves an optimal configuration as the parameter $q$ converges to $0$.
The proof is technically sophisticated: it relies on the characterization of stochastically stable states using techniques originally proposed by  P.~R.~Kumar and his coauthors~\cite{connors88,connors89,desai94} to study simulated annealing. In a previous version of this report~\cite{arxiv1} we used a different approach inspired by~\cite{young93} to prove the following weaker result: it is not possible to replace a single content at  one  of  the  caches  and  still  improve  the  performance metric of interest.

In order to illustrate the flexibility of~\qlrud, we show in Sect.~\ref{s:case_studies} how to particularize the policy for two specific performance metrics, i.e.,~the hit rate and the retrieval delay under CoMP. While our theoretical guarantees hold only asymptotically, numerical results show that \qlrud{} with $q\in [0.01,0.1]$ already approaches the performance of the {offline} allocation obtained through greedy, which, while not provably optimal, is the best baseline we can compare to. Note that the greedy algorithm requires complete knowledge of network topology, transmission characteristics, and request process, while \qlrud{} is a reactive policy that relies only on a noisy estimation of the marginal benefit deriving from a local copy. 

We remark that the goal of \qlrud{} and this paper is not to propose ``the best'' policy for \emph{any} scenario with ``coupled'' caches, but rather a simple and easily customizable policy framework with provable theoretical properties.
Currently, new caching policies designed for a particular scenario/metric are often compared with classic policies like \lru{} or \lfu{} or  the more recent \lruone{} and \lruall{}. This comparison appears to be quite unfair, given that these policies 1) ignore or only partially take into account the potential advantage of coordinated content allocations and 2) all target the hit-rate as performance metric. \qlrud{} may be a valid reference point, while being simple to implement. A Swiss-army knife is a very helpful object to carry around, even if each of its tools may not be the best one to accomplish its specific task.

\section{Network model}
\label{s:network_model}
We consider a set of $B$ base stations (BSs) arbitrarily located in a given region $R \subseteq \mathbb R^2$, each equipped with a local cache with size $C$.
Users request contents from a finite catalogue of size $F$.
{Given a content $f$, 
a specific allocation of  its copies  across the caches is specified by the vector $\x_f = (x^{(1)}_f, x^{(2)}_f,\dots, x^{(B)}_f)$,} where $x^{(b)}_f = 1$ (resp.~$x^{(b)}_f=0$) indicates that a copy of $f$ is present (resp.~absent) at BS $b$.
Let $\mathbf e^{(b)}$ be the vector with a $1$ in position $b$ and all other components equal to $0$. We write $\x_f \oplus \mathbf e^{(b)}$ to indicate a new cache configuration where a copy of content $f$ is added at base station $b$, 
if not already present {(i.e., $\x_f\oplus e^{(b)}=\x_f$  whenever $x_f^{(b)}=1$)}.
 Similarly,  $\x_f \ominus \mathbf e^{(b)}$ indicates a new allocation where there is no  copy of content $f$ at $b$ { ($\x_f \ominus \mathbf e^{(b)}=\x_f $ whenever $x_f^{(b)}=0$)}.
{Finally, we denote by $\X_f(t)=\left(X^{(1)}_f(t), \dots, X^{(B)}_f(t)\right)$, the specific content $f$ configuration  at time $t$.}

When user $u$ requests and receives content $f$, some network stakeholder
 achieves a gain  that we assume to depend on user $u$, content $f$ and the current allocation of content $f$ copies ($\X_f(t)$). We denote the gain as $g_f(\X_f(t),u)$. For example, if the key actor is the content server, $g_f(\X_f(t),u)$ could be the indicator function denoting if $u$ can retrieve the content from one of the local caches (reducing the load on the server). If it is the network service provider, $g_f(\X_f(t),u)$ could be the number of bytes caching prevents from traversing bottleneck links. Finally, if it is the user, $g_f(\X_f(t),u)$ could be the delay reduction achieved through the local copies. We consider that $g_f(\mathbf 0,u)=0$, i.e.,~if there is no copy of content $f$, the gain is zero.

The gain $g_f(\X_f(t),u)$ may be a random variable. For example, it may depend on the instantaneous characteristics of the wireless channels, or on some  user's random choice like the BS from which the file will be downloaded. We assume that, conditionally on the network status  $\X_f(t)$ and the user  $u$,
 these random variables are independent from one request to the other and are identically distributed with expected value $\E[g_f(\X_f(t),u)]$.

Our theoretical results hold under a stationary request process. In particular, we consider two settings. In the first one, there is a finite set of $U$ users located at specific positions. Each user $u$ requests the different contents according to independent Poisson process with rates $\lambda_{f,u}$ for $f \in \{1, 2, \dots, F\}$. The total expected gain per time unit from a given placement $\x_f$ is 
\begin{equation}
	G_f(\x_f) = \sum_{u=1}^U \lambda_{f,u} \E\left[g_f(\x_f,u)\right].
\end{equation}
In the second setting, a potentially unbounded number of users are spread over the region $R$ according to a Poisson point process with density $\mu()$.
Users are indistinguishable but for their position $\mathbf r$. In particular, a user $u$ in $\mathbf r$  generates a Poisson request process with rate $\lambda_f(\mathbf r)$ and experiences a gain $g_f(\x_f,\mathbf r)$.
The total expected gain from a given placement of content~$f$ copies is in this case
\begin{equation}
	G_f(\x_f) = \int_R \lambda_{f}(\mathbf r) \E\left[g_f(\x_f,\mathbf r)\right] \mu(\mathbf r) \textrm d \mathbf r.
\end{equation}
{We observe that $G_f(\cdot)$ is non negative and non-decreasing in the sense that $G_f(\x_f \oplus \mathbf e^{(b)})\ge G_f(\x_f)$, for each $\x_f$ and each $b$.}

In what follows, we will refer to the marginal gain from a copy at base station $b$. When the set of users is finite, we define the following quantities, respectively for a given user and for the whole network:
\begin{align}
&\Delta g_f^{(b)}(\x_f, u) \triangleq   g_f(\x_f,u) - g_f(\x_f \ominus \mathbf e^{(b)},u),\\
& \Delta G_f^{(b)}(\x_f) \triangleq  G_f(\x_f) - G_f(\x_f \ominus \mathbf e^{(b)}) \label{e:deltaGb}
\end{align}
{$\Delta g_f^{(b)}(\x_f, u)$ represents the cost reduction observed by user $u$ when the system 
moves from state $\x_f \mathbf e^{(b)}$ to state $ \x_f$. $\Delta G_f^{(b)}(\x_f) $ represents the average cost reduction when the system 
moves from state $ \x_f \ominus  \mathbf  e^{(b)}$ to state~$\x_f$.} 
It is possible to definite similarly $\Delta g_f^{(b)}(\x_f, r)$ when users' requests are characterized by a density over the region $R$.  In what follows, we will usually refer to the case of a finite set of users, but all results hold in both scenarios.

We would like our dynamic policy to converge to a content placement that maximizes the total expected gain, i.e.,
\begin{align}
\label{e:static_opt_gen}
& \maxim_{\x_1, \x_2, \dots, \x_F}
& & G(\x) \triangleq \sum_{f=1}^F  G_f(\x_f) \\ 
& \text{subject to}
& & \sum_{f=1}^F x_f^{(b)} =C  \;\;\; \forall b = 1, \ldots, B,\nonumber\\ 
& & & x_f^{(b)} \in \{0,1\} \;\;\; \forall f =1, \ldots, F, \nonumber\\
& & & \;\;\;\;\;\;\;\;\;\;\;\; \;\;\; \;\;\;\;\;\;\;\forall b = 1, \ldots, B. \nonumber 
\end{align}
even in the absence of a priori knowledge about the request process.
In the three specific examples we have mentioned above, solving problem~\eqref{e:static_opt_gen} respectively corresponds to 1) maximize the hit ratio, 2) minimize the network traffic, and 3) minimize the retrieval time. This problem is in general NP-hard, even in the case of the simple hit ratio metric~\cite{shanmugam13}.
{
Note also that it is possible to define opportunely the gain function to take into account a notion of fairness across contents, for example to determine a weighted $\alpha$-fair cache allocation~\cite{kelly14stochastic_networks}.
}

\section{\qlrud}
\label{s:operation}

We describe here how our system operates and the specific caching policy we propose to approach the solution of Problem~\eqref{e:static_opt_gen}. 

When user $u$ has a request for content $f$, it broadcasts an inquiry message to the set of BSs ($I_u$) it can communicate with.
 The subset ($J_{u,f}$) of those BSs that have the content $f$ stored locally declare their availability to user~$u$. If no local copy is available, the user sends the request to one of the BSs in $I_u$, which will need to retrieve it from the content provider.\footnote{
	This two-step procedure introduces some additional delay, but this is inevitable in any femtocaching scheme where the BSs need to coordinate to serve the content.
} If a local copy is available ($J_{u,f}\neq \emptyset$) and only point-to-point transmissions are possible, the user sends an explicit request to download it to one of the BSs in $J_{u,f}$. Different user criteria  can be defined to select the BS {in $J_{u,f}$} to download from (e.g., SNR, or pre-assigned priority list~\cite{LTE-book}). 
If CoMP techniques are supported, then all the BSs in $J_{u,f}$ coordinate to jointly transmit the content to the user.

Our policy \qlrud{} works as follows. Each BS $b$ with a local copy ($b \in J_{u,f}$) moves the content to the front of the cache with probability proportional to the marginal gain due to the local copy, i.e.,
\begin{equation}
\label{e:update}
p_f^{(b)}(u) = \beta \Delta g_f^{(b)}(\X_f(t),u),
\end{equation}
where the constant {$\beta\le \left(\max_{u, b, \x_f} \Delta g_f^{(b)}(\x_f,u)\right)^{-1}$ guarantees that the RHS of above equation is 
always in $[0,1]$}.
At least one of the BSs without the content (i.e., those in $I_{u,f} \setminus J_{u,f}$) {will store} an additional copy of $f$ with probability
\begin{equation}
\label{e:miss}
q_f^{(b)}(u) =  q^{(b)} \delta \Delta g_f^{(b)}(\X_f(t) \oplus \mathbf e^{(b)},u),
\end{equation}
where $\delta$ plays the same role of $\beta$ above and $q^{(b)}$ is a dimensionless parameter in $(0,1]$.

{
Some information about the local neighborhood (e.g., how many additional copies of the content are stored at close-by caches also serving that user) may be needed to compute the marginal gains in \eqref{e:update} and \eqref{e:miss}. Such information, however, is limited, and can be piggybacked on existing messages the user sends to query such caches, or even on channel estimates messages mobile devices regularly send to nearby BSs. In Sect.~\ref{s:case_studies} we detail what information needs to be exchanged when the system aims to maximize the hit rate or minimize the delay.
}

We are going to prove that \qlrud{} is asymptotically optimal when the values $q^{(b)}$ converge to $0$. 
This result holds under different variants for~\eqref{e:update} and~\eqref{e:miss}.
First, as it will be clear from the discussion in the following section, 
our optimality result depends on $\E[p_f^{(b)}(u)]$ being proportional to $\E[\Delta g_f^{(b)}(\X_f(t),u)]$. Then it is possible to replace $g_f^{(b)}(\X_f(t),u)$ in~\eqref{e:update} with any other unbiased estimator of $\E[g_f^{(b)}(\X_f(t),u)]$.
We are going to show an example when this is useful in~Sect.~\ref{s:case_studies}. 
 {upon a favourable random outcome, content $f$ can be retrieved
simultaneously by}
any number ($>0$) of BSs in $I_{u,f} \setminus J_{u,f}$ and the probability {$q_f^{(b)}(u)$ could be simply   chosen 
equal to $q^{(b)}$, i.e., made independent of the caching allocation.} 
We propose~\eqref{e:miss} because this rule is more likely to add copies that bring a large benefit $\Delta g_f^{(b)}(\X_f(t) \oplus \mathbf e^{(b)},u)$. This choice likely improves convergence speed, and then the performance in non-stationary popularity environments.

\section{Optimality of \qlrud}
\label{s:optimality}

We are going to prove that \qlrud{} achieves a locally optimal configuration when the values $q^{(b)}$ vanish. The result relies on two approximations: the usual characteristic time approximation (CTA) for caching policies (also known as Che's approximation)~\cite{fagin77,che02} and the new exponentialization approximation (EA) for networks of interacting caches originally proposed in~\cite{leonardi18jsac}.
The main results of this paper is the following:
 \begin{prop}
\label{p:qlrud_convergence_general} \textbf{[loose statement]}
Under characteristic time and exponentialization approximations, a spatial network of \qlrud{} caches asymptotically achieves an optimal caching configuration when $q^{(b)}$ vanish.
 \end{prop}

Before moving to the detailed proof, we provide some intuition about why this result holds. We observe that, as $q^{(b)}$ converges to $0$, cache $b$ exhibits two different dynamics with very different timescales: the \emph{insertion of new contents} tends to happen more and more rarely ($q_f^{(b)}(u)$ converges to $0$), while the frequency of \emph{position updates} for files already in the cache is unchanged ($p_f^{(b)}(u)$ does not depend on $q^{(b)}$). A file $f$ at cache $b$ is moved to the front with a probability proportional to $\Delta g_f^{(b)}(\X_f,u)$, i.e., proportional to how much the file contributes to improve the performance metric of interest. This is a very noisy signal: upon a given request, the file is moved to the front or not. At the same time, as $q$ converges to $0$, more and more moves-to-the-front occur between any two file evictions. The expected number of moves-to-the-front file $f$ experiences is proportional to 1) how often it is requested ($\lambda_{f,u}$) and 2) how likely it is to be moved to the front upon a request ($p_f^{(b)}(u)$). Overall, the expected number of moves is proportional to $\sum_u \lambda_{f,u}  \E\!\left[\Delta g_f^{(h)}(\X_f,u)\right]$, i.e.,~its contribution to the expected gain. By the law of large numbers, the random number of moves-to-the-front will be close to its expected value and it becomes likely that the least valuable file in the cache occupies the last position. 
We can then think that, when a new file is inserted in the cache, it will replace the file that contributes the least to the expected gain. \qlrud{} then behaves as 
a greedy algorithm that, driven by the request process, {replaces the least useful file in the cache at each insertion}, until it reaches a  maximum. 

\subsection{Characteristic Time Approximation} 
{In this section we focus on a single cache (i.e., one base station in isolation), or equivalently on a cache $b$ in a network of $B$ non-overlapping cells.}

{This is a  standard approximation for a cache in isolation}, and one of the most effective approximate approaches for analysis of caching systems. CTA was first introduced (and analytically justified) in~\cite{fagin77} and later rediscovered in~\cite{che02}. It was originally proposed for \lru{} under the IRM request process, and it has been later extended to different caching policies and different requests processes~\cite{garetto16,garetto15}. 

The characteristic time $T_c^{(b)}$ is the time a given content spends in the cache since its insertion until its eviction in absence of any request for it. In general,  {this quantity depends  in a complex way on the dynamics of other contents requests. Instead, the CTA assumes that $T_c^{(b)}$} 
 is  a random variable independent  from other contents dynamics and
with an assigned distribution (the same for every content). This assumption makes it possible to decouple the dynamics of the different contents: upon a miss for content $f$, the content is retrieved and a timer with random value $T_c^{(b)}$ is generated. When the timer expires, the content is evicted from the cache. 

Cache policies differ in {\it i)} the distribution of $T_c^{(b)}$ and {\it ii)} what happens to the timer upon a hit. For example, $T_c^{(b)}$ is a constant under \lru, \qlru, \twolru, and \fifo{} and exponentially distributed under \random. Upon a hit, the timer is renewed under 
\lru, \qlru, and \twolru, but not under \fifo{} and \random. {In what follows we will only consider policies for which $T_c^{(b)}$ is a constant.}
 Under CTA, the instantaneous cache occupancy can violate the hard buffer constraint.
{{The value of $T_c^{(b)}$ is obtained by imposing  the expected occupancy to be equal to the buffer size:
\begin{equation}
\label{e:che_single_cache}
\sum_{f=1}^F \pi_f^{(b)} = C
\end{equation}
where $\pi_f^{(b)}$ denotes the probability that content $f$ is in cache $b$. Its expression as a function of $T_c^{(b)}$ depends on the specific caching policy~\cite{garetto16}.}
Despite its simplicity, CTA  was shown  to provide asymptotically exact predictions for a single \lru{} cache under IRM  as the cache size grows large~\cite{fagin77,Jele99,fricker2012}.

Once inserted in the cache, a given content $f$ will sojourn in the cache for a random amount of time $T_{S,f}^{(b)}$, independently from the dynamics of other contents. $T_{S,f}^{(b)}$ can be characterized for the different policies. In particular, if the timer is renewed upon a hit, we have:
{
\begin{equation}\label{sojourn-time-struct}
T_{S,f}^{(b)}= \sum_{k=1}^{\infty} Y_k \mathrm{1}_{\{Y_1< T_c^{(b)}, \ldots, Y_k<T_c^{(b)} \}}+ T_c^{(b)}=  \sum_{k=1}^M Y_k+ T_c^{(b)},
\end{equation} 
where $M \in \{0,1, \dots\}$ is the number of consecutive hits following a miss, and \changes{$Y_k$ is the time interval between the $k$-th request following a miss and the previous content request}. 
}

We want to compute the expected value of $T_{S,f}^{(b)}$ that we denote as $1/\nu_f^{(b)}$.
When the number of users is finite, requests for content $f$ from user~$u$ arrive according to a Poisson process with rate $\lambda_{f,u}$. The time instants at which content $f$ is moved to the front are generated by  thinning this Poisson process with probability $\beta \E[\Delta g_f^{(b)}(u)]$.\footnote{
	Here we simply write $\Delta g_f^{(b)}(u)$ instead of $\Delta g_f^{(b)}(\X_f^{(b)},u)$, because we are considering a single cache. Similary, we write $\Delta G_f^{(b)}$, instead of $\Delta G_f^{(b)}(\X_f(t))$.
} The resulting sequence is then also a Poisson process with rate $\lambda_{f,u} \beta \E[\Delta g_f^{(b)}(u)]$. 
Finally, as request processes from different users are independent, the aggregate cache updates due to all users is a Poisson process with rate 
$$\beta \sum_{u=1}^U \lambda_{f,u} \E[\Delta g_f^{(b)}(u)] = \beta \Delta G_f^{(b)}.$$
The same result holds when we consider a density of requests over the region $R$.

As the aggregate cache updates follow a Poisson process with rate $\beta \Delta G_f^{(b)}$,  $\{Y_k\}$ are i.i.d. truncated exponential random variables with rate $\beta \Delta G_f^{(b)}$ over the interval $[0,T_c^{(b)}]$ and their expected value is 
\[  \E[Y_k]= \frac{1}{\beta \Delta G_f^{(b)}} - \frac{T_c^{(b)} } {e^{\beta \Delta G_f^{(b)} T_c^{(b)}} -1} .\]
Moreover, the probability that no update occurs during a time interval of length $T^{(b)}_c$ is $e^{-\beta \Delta G_f^{(b)} T^{(b)}_c}$. Then $M$ is distributed as a geometric random variable with values $\{0, 1, \dots\}$ with expected value
\[\E[M]=\frac{1-e^{-\beta \Delta G_f^{(b)} T_c^{(b)}}}{e^{-\beta \Delta G_f^{(b)} T_c^{(b)}}}= e^{\beta \Delta G_f^{(b)} T_c^{(b)}} -1.\] 
{Since $M$ is clearly a stopping point for the sequence $\{ Y_k\}_k$,} we can then apply Wald's Lemma to \eqref{sojourn-time-struct} obtaining:
\begin{align}
\label{e:rate}
\nu_f^{(b)}& \triangleq  \frac{1}{\mathbb{E}[T_{S,f}^{(b)}]} =   \frac{1}{\mathbb{E}[Y_1] \;\mathbb{E}[M]  +T_c^{(b)}}\nonumber\\
 	& = \frac{\beta \Delta G_f^{(b)}}{e^{\beta \Delta G_f^{(b)} T_c^{(b)}}-1}.
\end{align}

\subsection{Exponentialization Approximation}
We consider now the case when $B$ cells may overlap.  
The sojourn time of content $f$ inserted at time $t$ in cache~$b$ will now depend on the whole state vector $\X_f(\tau)$ for $\tau \ge t$ (until the content is not evicted), because the content is updated with probability~\eqref{e:update} depending on the marginal gain of the copy (and then on $\X_f(\tau)$). EA consists to assume that the stochastic process $\X_f(t)$ is a continuous-time Markov chain.  For each $f$ and $b$ the transition rate $\nu_f^{(b)}$ from state $\X_f(t)=(x_f^{(b)}=1, \x_f^{(-b)})$ to $(x_f^{(b)}=0, \x_f^{(-b)})$ is given by \eqref{e:rate} with $\Delta G_f^{(b)}$ replaced by $\Delta G^{(b)}_f(\X_f(t))$.
EA replaces then the original stochastic process, whose analysis is extremely difficult, {with a set of MCs $\X_f(t)$, for $f=1, \dots, F$, which are only coupled through the characteristic times $T_c^{(b)}$ at the BSs.} Reference~\cite{leonardi18jsac} shows that this has no impact  on  any system metric that depends only on the stationary distribution in the following cases:
\begin{enumerate}
	\item isolated caches,
	\item caches using \random{} policy,  
	\item caches using \fifo{} policy as far as the resulting Markov Chain $\X_f(t)$ is reversible.
\end{enumerate}
Numerical results in \cite{leonardi18jsac} show that the approximation is practically very accurate also in more general cases.

{Similarly to what done for a single cache, we can determine the values $T_c^{(b)}$ at each cache, by imposing that:
\begin{equation}
	\label{e:che_multi_cache}
	\sum_{f=1}^F\sum_{\x_f \in \{0,1\}^B} x_f^{(b)} \pi_f(\x_f) = C,
\end{equation}
where $\pi_f(\x_f)$ denotes the stationary probability that MC $\X_f(t)$ is in state $\x_f$.
}
\subsection{Transition rates of the continuous time Markov Chain as $q$ vanishes}

For a given content $f$, let  $\x_f$ and $\y_f$  be two possible states of the MC $\X_f(t)$. We write 
  $\x_f < \y_f$    whenever  $x_f^{(b)} \le  y_f^{(b)}$ for each $b$ and there is at least one $b_0$ such that $x_f^{(b_0)} <  y_f^{(b_0)}$, and we say that $\bold{y}_f$ is an \emph{ancestor} of $\x_f$, and $\x_f$ is a \emph{descendant} of $\bold{y}_f$. Furthermore we denote by  
 $|\x_f|=\sum_b x_f^{(b)}$ the number of  copies of  content $f$ stored in state $\x_f$, and we call it the weight of the state $\x_f$. If $\x_f < \y_f$ and $|\x_f|=|\y_f|-1$, we say that  $\y_f$ is a \emph{parent} of $\x_f$ and $\x_f$ is a \emph{child} of $\y_f$.

 Now  observe that, by construction, transition rates in the MC are different from 0 only between pair of states  $\x_f$ and $ \y_f$,  
   such that   $\x_f < \y_f$ or $\y_f < \x_f$.  The transition $\x_f \to \y_f$ is called an \emph{upward} transition, while  $\y_f \to \x_f$ is called a \emph{downward} transition.
   
 A downward transition can only occur from a parent to a child ($|\x_f |= |\y_f |-1$). Let  $b_0$ be the index such that    $x_f^{(b_0)}<y_f^{(b_0)}$. We have that the downward rate is
\begin{equation}
\label{e:downward_rate}
\rho_{[\y_f \to \x_f]} = \nu_f^{(b_0)}(\y_f) = \frac{\beta \Delta G_f^{(b_0)}}{e^{\beta \Delta G_f^{(b_0)}(\y_f) T_c^{(b_0)}}-1}.
\end{equation}
Upward transitions can occur to states that are ancestors. The exact transition rate between state $\x_f$ and state $\y_f$ with $\x_f < \y_f$ can have a quite complex expression, because it depends on the joint decisions of the BSs in $I_{u,f}\setminus J_{u,f}$. Luckily, for our analysis, we are only interested in how this rate depends on $q$, when $q$ converges to $0$. We use the symbol $\propto$ to indicate that  two quantities are asymptotically proportional for small $q$, i.e., $f(q) \propto g(q)$ if and only if there exists a strictly positive constant $a$ such that $\lim_{q \to 0} f(q)/g(q)=a$. If $a=1$, then we write $f(q)\sim g(q)$ following Bachmann-Landau notation.

Upon a request for $f$, a transition $ \x_f\to \y_f$ occurs, if $|\y_f| - |\x_f|$ BSs  independently store, each with probability proportional to its parameter $q^{(b)}$,  an additional copy of the content $f$ in their local cache. It follows that:
\begin{equation}
\label{e:upward_rate}
\rho_{[\x_f \to \y_f]} \propto \prod_{b | y_f^{(b)} - x_f^{(b)} = 1} q^{(b)}.
\end{equation}


Now, as $q^{(b)}$ converges to $0$,  for every $f$ every upward rate  $\rho_{[\x_f \to \y_f]}$ tends to 0. Therefore, the characteristic time of every cell $T_C^{(b)}$ must diverge. In fact, if it were not the case for a cache $b$,  none of the contents would be found in this cache  asymptotically, because upward rates would tend to zero, while downward rates would not. This would contradict the set of constraints~\eqref{e:che_multi_cache} imposed by the CTA.
Therefore necessarily $T_C^{(b)}$ diverges for every cell $b$.  More precisely, we must have
$T_C^{(b)}=\Theta(\log \frac{1}{q})$  at every cache, otherwise we fail to  meet~\eqref{e:che_multi_cache}.
In other words, there exist  positive constants $a_l^{(b)}$ and $a_u^{(b)}$, such that $T_C^{(b)}(q)/ \log (1/q)$ asymptotically belongs to $[a_l^{(b)},a_u^{(b)}]$. 
{
Given that the behavior $T_C^{(b)}(q)/ \log (1/q)$   is expected to be smooth, we assume that there exist (potentially different) positive constants  $\gamma_b$ for all $b \in \{1, \dots, B\}$ such that $T_C^{(b)}(q)\sim \frac{1}{\beta \gamma_{b}}(\log \frac{1}{q})$ and $\frac{1}{\beta \gamma_{b} }\in [a_l^{(b)},a_u^{(b)}]$.}

Now, we consider that BS $b$ employs $q^{(b)}= q^{\gamma_b}$. This choice makes the characteristic time scale in the same way at each cache: $T_C^{(b)}(q^{(b)}) \sim \frac{1}{\beta}\log \frac{1}{q}$. From this result and~\eqref{e:downward_rate}, it follows that a downward transition  from a parent $\y_f$ to a child $\x_f=\y_f \ominus \mathbf e^{(b_0)}$ occurs with rate 
\[ \rho_{[\y_f \to \x_f]} \propto q^{\Delta G_f^{(b_0)}(\y_f) }. \]

The following lemma summarises the results of this section.
{
\begin{lem}
\label{l:asymptotic_rates}
Consider two neighbouring states $\x_f$ and $\y_f$ with $\x_f < \y_f$ and the set of positive constants $\{\gamma_b, b=1, \dots, B\}$, such that $T_c^{(b)}(q^{(b)}) \sim \frac{1}{\beta \gamma_b} \log \frac{1}{q}$. If $q^{(b)} = q^{\gamma_b}$ then
\[ \rho_{[\x_f \to \y_f]} \propto q^{\vgamma^\intercal \left(\y_f - \x_f\right)}, \]
 if $ \x_f  =  \y_f \ominus \mathbf e^{(b_0)} $, then  
\[ \rho_{[\y_f \to \x_f]} \propto q^{\Delta G_f^{(b_0)}(\y_f) }. \]
\end{lem}
}
From now on we will assume that $q^{(b)} = q^{\gamma_b}$. 

For each possible transition, we define its \emph{direct resistance} to be the exponent of the parameter $q$, then $r_f(\x_f,\y_f)=\vgamma^\intercal \left(\y - \x\right)$, $r_f(\y_f,\x_f)= \Delta G_f^{(b_0)}(\y_f)$ and $r_f(\x_f,\x_f)=0$. 
Observe that the higher the resistance, the less likely the corresponding transition.

\subsection{Stochastically stable states}
 In this section, we first introduce  the key concept of stochastically stable states, in which, as $q$ converges to $0$, the system gets trapped. Then,   we
 provide a characterization of   stochastically stable states (Corollary~\ref{c:stochastically_stable}), which  will be useful in Sect.~\ref{s:opt_proof} to prove that  they correspond to optimal configurations.

{
We consider the discrete time MC $\hat{\bold{X}}_f(k)$, obtained sampling the continuous time MC $\bold{X}_f(t)$ with a period $\tau>0$, i.e., $\hat{\bold{X}}_f(k)=\bold{X}_f(k \tau)$. 
}
Let $P_{f,q}$ denote the transition probability matrix of $\hat{\bold{X}}_f(k)$. For $q=0$, the set of contents in the cache does not change, each state is an absorbing one and any probability distribution is a stationary probability distribution for $P_{f,0}$. We are rather interested in the asymptotic behaviour of the MC when $q$ converges to $0$.
For $q>0$ the MC is finite, irreducible,\footnote{
	This is guaranteed if insertion probabilities in \eqref{e:miss} are positive. In some specific settings, it may be $\Delta g_f^{(b)}(\X_f(t) \oplus \mathbf e^{(b)},u)=0$ for each $u$. We can then consider $q_f^{(b)}(u) =  q \gamma \max(\Delta g_f^{(b)}(\X_f(t) \oplus \mathbf e^{(b)},u),\epsilon)$ with $\epsilon>0$, or simply $q_f^{(b)}(u) =q$.
}
 and aperiodic and then admits a unique stationary probability $\bold \pi_{f,q}$.
\begin{defn}\label{d:stocstable}
 A state $\x_f$ is called stochastically stable if $\lim_{q \to 0 } \bold \pi_{f,q}(\x_f) >0$. 
\end{defn}

We are going to characterize such states. {
The set of possible transitions of $\hat{\bold{X}}_f(k)$ is in general larger than the set of possible transitions of $\bold{X}_f(t)$, as multiple transitions of $\bold{X}_f(t)$ can occur during the period $\tau$. 
For example, $\bold{X}_f(t)$ cannot move directly from $\x_f$ to $\x''_f = \x_f \ominus \mathbf e^{(b_1)}\ominus \mathbf e^{(b_2)}$ with $|\x_f''| = |\x_f|-2$, but during the interval $\tau$ it could move from $\x_f$ to $\x'_f = \x_f \ominus \mathbf e^{(b_1)}$ and then from $\x'_f$ to $\x''_f$. 
The transition $\x_f \to \x''_f$ is then possible for  $\hat{\bold{X}}_f(k)$. At the same time, for small values of $\tau$ and of $q$, the probability of a direct transition $\x_f \to \x'_f$ is proportional to $q^{r(\x_f, \x_f')}  \tau + o\left(q^{r(\x_f, \x_f')}  \right) + o(\tau)$, but the probability of a combined transition $\x_f \to \x'_f \to \x''_f$ is smaller than $q^{r(\x_f, \x_f')+r(\x'_f, \x''_f)}  \tau^2 + o\left(q^{r(\x_f, \x_f')}  \right) + o\left(q^{r(\x'_f, \x''_f)}  \right) + o(\tau)$. These transitions may be neglected as their transition probabilities are $o(\tau)$ and their equivalent resistance is equal to the sum of the direct transitions they are composed by. We can then restrict ourself to consider the transitions in $\bold{X}_f(t)$.}

{
Each MC $\hat{\X}_f(k)$ has then transition rates proportional to a power of $0<q<1$, i.e.~$P_{f,q}(\x_f,\x'_f) \propto q^{r_f(\x_f,\x'_f)}$.\footnote{We omit from now on, the proportionality to $\tau$.} These MCs were studies in a series of papers~\cite{connors88,connors89,desai94} by P.~R.~Kumar and his coauthors, because of their relation with the MCs that appear in simulated annealing problems, where $r_f(\x_f, \x'_f) = \max(C(\x'_f)-C(\x_f),0)$ and $C(\x_f)$ is a cost function we want to minimize. We list as lemmas three results from those papers we are going to use.

Consider a weighted graph $\mathcal G_f$, whose nodes are the possible states $\x_f \in \{0,1\}^B$ and edges indicate possible direct transitions and have a weight equal to the corresponding resistance. Given an in-tree $\mathcal T(\x_f)$ in $\mathcal G_f$ routed in $\x_f$, we denote by $r_f(\mathcal T(\x_f))$ the resistance of the in-tree, i.e., the sum of all resistances of the edges of $\mathcal T(\x_f)$. We also denote by $\mathfrak T(\x_f)$ the set of all in-trees routed in state $\x_f$. Finally, we denote by $r_f(\x_f)$ the resistance of the minimum weight in-tree (or anti-arborescence) in $\mathcal G_f$ rooted to $\x_f$, i.e.,
\[r_f(\x_f) \triangleq \min_{\mathcal T \in \mathfrak T(\x_f) }  r_f(\mathcal T).\]
Intuitively, the resistance of a state is a measure of the general difficulty to reach state $\x_f$ from all other nodes.
A consequence of the Markov chain tree theorem (see for example~\cite{anantharam89}) is that
\begin{lem}
\label{l:stationary_distribution}
\cite[Lemma~1]{desai94} The stationary probabilities of the MC $X_{f,q}(k)$ have the following expression
\[\pi_{f,q}(\x_f) \propto q^{r_f(\x_f) - \underset{\x'_f}{\min } \; r_f(\x'_f)}.\]
\end{lem}
A consequence of Lemma~\ref{l:stationary_distribution} is that the stochastically stable states are those with minimal resistance.

Consider the following system of \emph{modified balance equations} in the variables $\nu_f(\x)$:
\begin{equation}
\label{e:balance_eqs}
\left\{
\begin{aligned}
&  \underset{\x_f \in A, \z_f \in A^c}{\max}   \nu_f(\x_f)    -   r_f(\x_f,\z_f) \\
  &\phantom{===}  =  \underset{\x_f \in A, \z_f \in A^c}{\max} \nu_f(\z_f) - r_f(\z_f,\x_f),\;  \\
  &\phantom{=====} \forall A \subset \{0,1\}^B\\
&  \underset{\x_f \in \{0,1\}^B}{\max}   \nu_f(\x_f)  = \sigma.
\end{aligned}
\right.
\end{equation}
\begin{lem}\cite[Theorem~3]{connors89}
\label{l:balance_eqs}
For each $\sigma$, the system~\eqref{e:balance_eqs} admits a unique solution. Solutions for different values of $\sigma$ are translates of each other.
\end{lem}
System~\eqref{e:balance_eqs} implicitly determines the set of stochastically stable states:
\begin{lem}
\label{l:balance_stoch_stable}
\cite[Theorem~4]{desai94}
Given $\{\nu_f(\x_f)\}$ the solution of system~\eqref{e:balance_eqs}, 
 it holds:
\[r_f(\x_f) - \min_{\x'_f} r_f(\x'_f) = \sigma - \nu_f(\x_f).\]
\end{lem}

In particular for our system, we can prove that 
\begin{lem}
\label{l:our_balance}
The function 
\[\phi_f(\x_{f})\triangleq G_f(\x_f) - \vgamma^\intercal \x_f\]
is a solution of system~\eqref{e:balance_eqs} (for a particular value of $\sigma$).
\end{lem}
The proof is in Appendix~\ref{a:our_balance}.

A consequence of Lemma~\ref{l:stationary_distribution}, Lemma~\ref{l:balance_eqs}, Lemma~\ref{l:balance_stoch_stable}, and Lemma~\ref{l:our_balance} is that
\begin{cor}
\label{c:stochastically_stable}
The set of stochastically stable states is the set of global maximizers of $\phi_f(\x_f)$.
\end{cor}
}

For each content $f$ we are then able to characterize which configurations are stochastically stable as $q$ converges to $0$.

\subsection{Optimality proof}
\label{s:opt_proof}
{
We now consider the continuous relaxation of the optimization problem~\eqref{e:static_opt_gen}:
\begin{align}
\label{e:relaxed_opt}
& \maxim_{\{\alpha_f(\x_f)\}}
& & \sum_{f=1}^F \sum_{\x_f \in \{0,1\}^B} \alpha_f(\x_f) G_f(\x_f) \\ 
& \text{subject to}
& & \sum_{f=1}^F \sum_{\x_f \in \{0,1\}^B} \alpha_f(\x_f) x_f^{(b)} =C, \;\;\; \forall b \in [B] \nonumber\\ 
& & & \sum_{\x_f \in \{0,1\}^B} \alpha_f(\x_f) = 1, \;\;\;\;\;\;\;\;\;\;\;\;\;\;\; \forall f \in [F] \nonumber\\
& & & \alpha_f(\x_f) \ge 0, \;\;\;\;\;\;\;\;\;\;\;\;\;\;\; \forall f \in [F], \forall b \in [B]. \nonumber 
\end{align}
The optimization problem~\eqref{e:static_opt_gen} corresponds to the particular case, where we require that, for each $f \in [F]$, there exists a single state $\x_f$ with $\alpha_f(\x_f) = 1$ and  $\alpha_f(\x'_f) = 0$ for each $\x'_f \neq \x_f$. As the feasible set of the relaxed problem~\eqref{e:relaxed_opt} includes the feasible set of problem~\eqref{e:static_opt_gen}, the optimum value of problem~\eqref{e:relaxed_opt}  is at least as large as the optimal value of problem~\eqref{e:static_opt_gen}.

Note how the capacity constraint in problem~\eqref{e:relaxed_opt} is similar to the relaxed constraint considered by the CTA (see~\eqref{e:che_multi_cache}). 
This suggests that the stationary probabilities $\pi_f(\x_f)$ will play the role of the coefficients $\alpha_f(\x_f)$.

}

Now we can state formally our result.
\begin{repprop}{p:qlrud_convergence_general}
Under characteristic time and exponentialization approximations, let $\{\gamma_b, b=1, \dots, B\}$ be the constants in~Lemma~\ref{l:asymptotic_rates}. Consider the spatial network of \qlrud{} caches, where cache $b$ selects the parameter $q^{(b)} = q^{\gamma_b}$. {As $q$ converges to $0$, the stationary probabilities $\{\pi_{f,q}(\x_f), f \in [F], \x_f \in \{0,1\}^B\}$ converge to an optimal solution of Problem~\eqref{e:relaxed_opt}.}
\end{repprop}
{
The proof is in Appendix~\ref{a:proof_qlrud}. It relies on the characterization of stochastically stable states in Corollary~\ref{c:stochastically_stable} and on studying problem~\eqref{e:relaxed_opt} using the method of Lagrange multipliers.
}

\section{Case studies}
\label{s:case_studies}
As we discussed, \qlrud{} can be made to optimize different utility functions $G_f(\cdot)$. In this section we illustrate two specific case studies: hit rate maximization, and delay minimization with CoMP techniques. We first describe what form the general \qlrud{}  assumes in these cases and then illustrate with some experiments the convergence result in Proposition~\ref{p:qlrud_convergence_general}. 

\subsection{Hit rate maximization}
The gain is simply $1$ from a hit and $0$ from a miss,~i.e.,
\[g_f(\X_f,u) = \mathbbm 1(J_{u,f} \neq \emptyset),\]
where $\mathbbm 1(\cdot)$ denotes the indicator function.
According to~\eqref{e:update} with $\beta =1$, each BS $b$ with a local copy ($b \in J_{u,f}$) moves the content to the front of the cache with probability
\begin{align*}
p_f^{(b)}(u) & = \Delta g_f^{(b)}(\X_f(t),u) \\
	&= \mathbbm 1(J_{u,f} \neq \emptyset) - \mathbbm 1(J_{u,f} \setminus \{b\} \neq \emptyset)\\
	& = 1 - \mathbbm 1(J_{u,f} \setminus \{b\} \neq \emptyset)\\
	& = \mathbbm 1(J_{u,f} \setminus \{b\} = \emptyset) = \mathbbm 1(J_{u,f}=  \{b\}),
\end{align*}
where the third equality is due to the fact that $b \in J_{u,f}$.
Similarly, from~\eqref{e:miss}, at least one of the BSs without the content (i.e., those in $I_{u,f} \setminus J_{u,f}$) decides to store an additional copy of $f$ with probability
\begin{align*}
q_f^{(b)}(u) & = q \mathbbm 1(J_{u,f} = \emptyset).
\end{align*}
The policy then works as follows. Upon a miss ($J_{u,f}=\emptyset)$, at least one cache decides to retrieve the content with probability $q$.
Upon a hit ($J_{u,f}\neq \emptyset$), the cache serving the content brings it to the front if and only if no other cache could have served it (i.e.,~$|J_{u,f}| = 1$). 

{Note that in order to compute $p_f^{(b)}$ and $q_f^{(b)}$ cache~$b$ simply needs to know the size of $J_{u,f}$. The system can then operate as follows: the user broadcasts a query for content $f$, discovers  $J_{u,f}$ (which BSs have a copy of the content) and piggyback this information when querying the specific BS from which it wants to retrieve the content}.

This policy is a slight extension of \qlrulazy{} proposed in~\cite{leonardi18jsac}. The only minor difference is that under \qlrulazy{} only one cache retrieves the contents upon a miss. \qlrud{} allows for some additional flexibility. In what follows, we consider that each cache decides independently to retrieve the copy (and then multiple copies of the same content can be retrieved).

\subsection{Delay minimization with CoMP}
Let $h_{b,u}$ denote the signal-to-noise ratio (SNR) of the wireless channel between BS $b$ and user $u$. We assume for simplicity that $\{h_{b,u}, b \in I_u\}$ are i.i.d. random variables with expected value $h$, and we consider $h_{b,u}=0$, when $u$ is not reachable by the BS $b$ ($b\notin I_u$). 
We consider BSs can employ a coordinated transmission technique. In particular, BSs in $J_{u,f}$ can cooperate to transmit the  file to $u$, and we assume they are able to achieve the aggregate channel capacity $C\!\left(\sum_{b \in J_{u,f}} h_{b,u}\right)\triangleq W \log_2(1+\sum_{b \in J_{u,f}} h_{b,u})$, where $W$ is the channel bandwidth \cite{tse2005fundamentals, ao15}.

Upon a miss, the content needs to be retrieved from a base station $b^* \in I_u$, selected uniformly at random, and then transmitted from $b^*$ to $u$.\footnote{
	It is possible to consider more complicated schemes, e.g.,~where the $u$ retrieves from the BS with the highest SNR.
}
The user then experiences a delay equal to the backhaul delay (denoted as $d_B$) plus the transmission delay $M/ C( h_{b^*,u})$, where $M$ is the size of the content.

Upon a hit, the delay is instead equal to 
\begin{align}
\frac{M}{C\!\left(\sum_{b \in J_{u,f}} h_{b,u}\right)} & = \frac{M}{C\!\left(\sum_{b \in I_{u}} h_{b,u} X_f^{(b)}\!(t)\right)}\\
	& = \frac{M}{C\!\left(\sum_{b} h_{b,u} X_f^{(b)}\!(t)\right)}
\end{align}
Summing up, the delay experienced by user $u$ requesting file $f$ is
\begin{align*}
	d_f(\X_f(t), u)& =
	\begin{cases}
		d_B + \frac{M}{C( h_{b^*,u})}, & \textrm{if }J_{u,f}= \emptyset,\\
		\frac{M}{C\left( \sum_{b} h_{b,u} X_f^{(b)}\!(t)\right)}, & \textrm{otherwise.}
	\end{cases}
\end{align*}

The total expected delay per request is then 
\begin{equation}
	D_f(\X_f(t)) = \sum_{u=1}^U \lambda_{f,u} \E\left[d_f(\X_f(t),u)\right],
\end{equation}
when the set of users is finite, and
\begin{equation}
	D_f(\X_f(t)) = \int_R \lambda_{f}(\mathbf r) \E\left[d_f(\X_f(t),\mathbf r)\right] \mu(\mathbf r) \textrm d \mathbf r,
\end{equation}
when a potentially unbounded set of users is distributed over the region (see Sect.~\ref{s:network_model}).

We want to minimize the delay $D_f(\x_f)$. In order to frame this goal according to our reference maximization problem~\eqref{e:static_opt_gen}, we can simply consider
$G_f(\x_f) \triangleq d_{\max} - D_f(\x_f)$, where $d_{\max}$ is a bound on the retrieval time, e.g.,~equal to the sum of the backhaul delay and the maximum delay on the transmission channel. Similarly, we consider $g_f(\x_f,u) \triangleq d_{\max} - d_f(\x_f,u)$. Note that 
\begin{align*}
& \Delta g_f^{(b)}(\x_f,u)     \\
	& =(d_{\max} - d_f^{(b)}(\x_f,u))\\
 	 & \;\;\;\;\; - (d_{\max} - d_f^{(b)}(\x_f\ominus \mathbf e^{(b)},u)) \nonumber\\
	& =  d_f^{(b)}(\x_f\ominus \mathbf e^{(b)},u)  - d_f^{(b)}(\x_f,u)\\
	& = \begin{cases}
		 d_B + \frac{ M}{c\left( h_{b^*,u}\right)} - \frac{ M}{c\left( h_{b,u}\right)},  \hspace{0.5cm}\textrm{        if }J_{u,f}= \{b\},\\
		\frac{ M}{C\left( \sum_{b'\neq b} h_{b',u} X_f^{(b')}\right)} - \frac{M}{C\left( \sum_{b'} h_{b',u} X_f^{(b')}\right)},  \hspace{0.2cm}\textrm{o/w.}
	\end{cases}
\end{align*}
Note that $d_{\max}$ cancels out and then the choice of its value is irrelevant for the  algorithm.

Remember from our discussion at the end of Sect.~\ref{s:operation} that it is possible to replace $\Delta g_f^{(b)}$ in~\eqref{e:update} with any other function with the same expected value. Given that $h_{b,u}$ and $h_{b^*,u}$ are identically distributed, we can have each BS $b$ with a local copy ($b \in J_{u,f}$) move the content to the front of the cache with probability
 \begin{align*}
& p_f^{(b)}(u)  = \\
	& = \begin{cases}
		\beta d_B,  \hspace{4.9cm}\textrm{        if }J_{u,f}= \{b\},\\
		\frac{\beta M}{C\left( \sum_{b'\neq b} h_{b',u} X_f^{(b')}\right)} - \frac{\beta M}{C\left( \sum_{b'} h_{b',u} X_f^{(b')}\right)},  \hspace{0.2cm}\textrm{o/w.}
	\end{cases}
\end{align*}
Similarly, from~\eqref{e:miss}, at least one of the BSs without the content (i.e., those in $I_{u,f} \setminus J_{u,f}$) decides if storing an additional copy of $f$ with probability
\begin{align*}
& q_f^{(b)}(u)  = \\
	& = \begin{cases}
		q \delta d_B,  \hspace{4.9cm}\textrm{        if }J_{u,f}= \emptyset,\\
		\frac{q \delta M}{C( \sum_{b'} h_{b',u} X_f^{(b')})} - \frac{q \delta M}{C( h_{b,u}+\sum_{b'} h_{b',u} X_f^{(b')}))},  \hspace{0.2cm}\textrm{o/w.}
	\end{cases}
\end{align*}
As above, we consider that each cache decides independently to retrieve an additional copy.

{Similarly to what discussed above, the user can piggyback to its request the measured SNRs values from all the BSs in its transmission range ($I_u$). This information allows BS $b$ to compute $p_f^{(b)}$ and  $q_f^{(b)}$.}

\begin{figure}[t]
   \centering
   \includegraphics[width=\myFigureScale\linewidth]{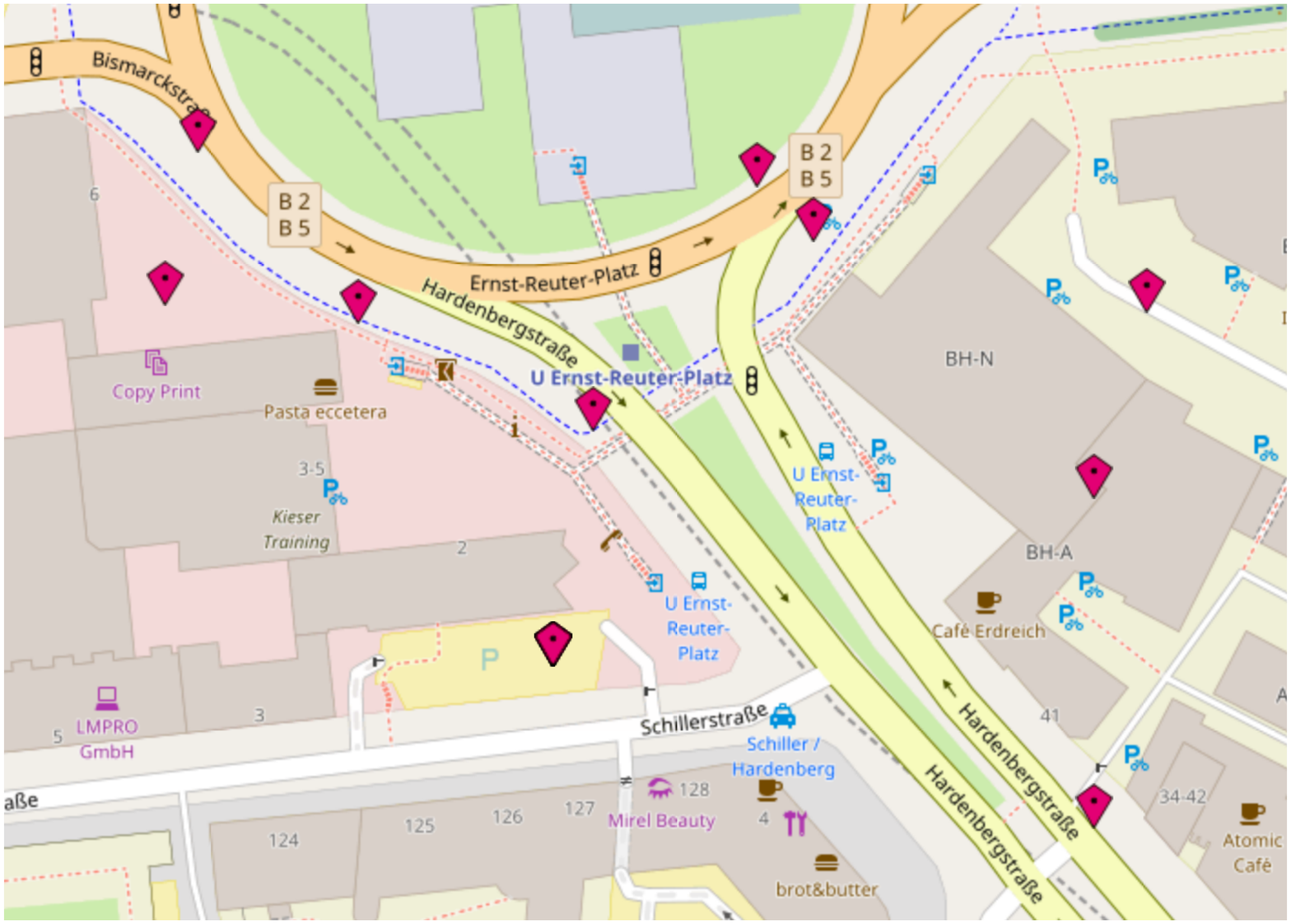}
   \caption{T-Mobile BS configuration in Berlin.} 
   \label{f:berlin}
\end{figure}

\subsection{Numerical Results}

{In our simulations} we consider a topology where $B=10$ base stations are located according to the positions of T-mobile base stations in Berlin extracted from~\cite{bs_dataset}. The BS locations are indicated in Fig.~\ref{f:berlin}. We assume their transmission range is 150m, and spatial user density to be homogeneous, so that each user on average is covered by 5.9 BSs. SNRs have constant values $h_{b,u}=10$dB, the channel bandwidth is $W=5.0$MHz,  and backhaul access delay is $d_B=100$ms. The catalog counts $F=10^6$ files with size $M=10^6$~bits, whose popularity distribution follows a Zipf law with exponent $\alpha=1.2$. Each BS has a local cache with capacity $C=100$~files{, unless otherwise stated.}

We show the performance of \qlrud, when it is configured to maximize the hit rate and when to minimize the delay rate. In the figures we refer to the two cases as \textproc{qLRU-$\Delta h$} and \textproc{qLRU-$\Delta d$}. For \textproc{qLRU-$\Delta d$},
{ we set $\beta$ and $\delta$ equal  to the minimum value that guarantees respectively $p_f^{(b)}(u)\le 1$ and $q_f^{(b)}(u)\le q$ for every possible state of the cache $\x_f$}.

We would like to compare their performance with the corresponding optimal {offline} allocations. Unfortunately, both corresponding optimization problems are NP-hard, but the greedy algorithm has a guaranteed $1/2$-approximation ratio for hit ratio maximization~\cite{shanmugam13} and for delay minimization~\cite{tuholukova17}.\footnote{
	Precisely, the greedy static allocation achieves at least $1/2$ of the delay savings achievable by the best possible static allocation.
} We then consider the corresponding {offline} allocations as baselines and denote them respectively as \greedyh{} and \greedyd. Note that the greedy algorithm requires complete knowledge of the network and of content popularities, while \qlrud{} has no such information.

{Additionally, we provide the results for the simulation of two other online policies: \qlru{} and \fifo{}. Both policies maintain the contents in the cache as an ordered list with insertions occurring at the front of the list and evictions at the rear.  In \qlru{}, the requested content is inserted with probability $q$ upon a miss and moved to the front upon a hit. Note that \qlru{} with $q=1$ coincides with \lru. In \fifo, the requested content is always inserted upon a miss and maintains its position upon a hit.} 

In all our experiments, policies' simulations have a warm up phase and a measurement phase each consisting of $10^8$ requests. 

\begin{figure}[h]
         \centering
         \includegraphics[width=0.42\textwidth]{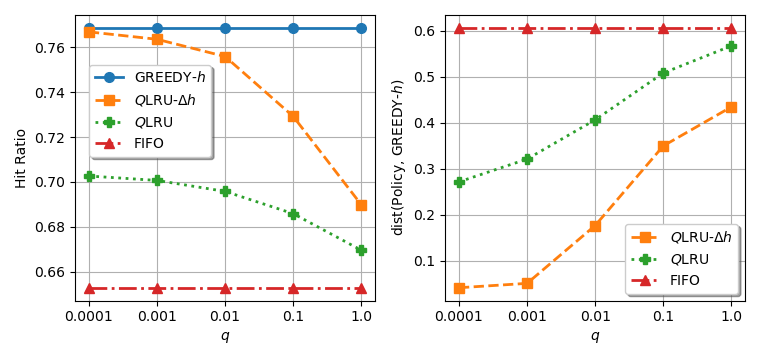}
        \caption{{Comparison of {online} policies and \greedyh: hit rate (left) and distance of their allocations (right) versus $q$.}}
         \label{fig:hit_rate}
\end{figure}

\begin{figure}[h]
         \centering
         \includegraphics[width=0.42\textwidth]{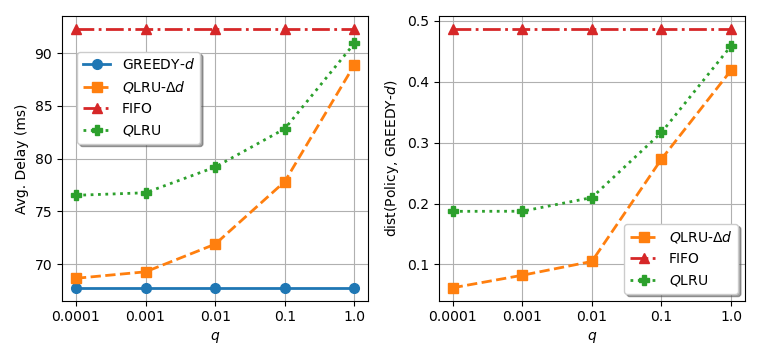}
        \caption{{Comparison of {online} policies and \greedyd: average delay (left) and distance of their allocations (right) versus $q$.}}
         \label{fig:delay}
\end{figure}

Figure~\ref{fig:hit_rate} (left) shows the hit rate achieved by \mbox{\greedyh} and by \qlrudh{} for different values of~$q$. As $q$ decreases, \qlrudh's hit rate converges to that of \greedyd. 
{The hit rate of \qlru{} also improves for smaller $q$. For a single cache, \qlru{} coincides with \qlrudh{} and it is then implicitly maximizing the hit rate when $q$ converges to $0$. But in a networked setting, the deployment of \qlru{} at each cache does not perform as well because each cache is myopically maximizing its own hit rate without taking into account the presence of the other ones. Instead, \qlrudh{} correctly takes into account the marginal contribution the cache can bring to the whole system. Finally, \fifo{} achieves the lowest hit rate as the sojourn time of each content inserted in the cache is roughly the same, independently from its popularity.}

We also compare how different the content allocations of~\qlrudh, \qlru, and \fifo{} are from the allocation of~\mbox{\greedyh}. To this purpose, we define the \emph{occupancy vector}, whose component $i$ contains the number of copies of  content $i$ present in the network averaged during the measurement phase. We then compute the cosine distance\footnote{
The cosine distance between vectors $u$ and $v$ is given by $\text{dist}(u,v) = 1 - \frac{\langle u, v \rangle}{\lVert u \rVert_2 \lVert v \rVert_2}$, where $\langle\cdot,\cdot\rangle$ denotes the inner product.
} of the occupancy vectors of  the specific {online} policy and \greedyh.
{Figure~\ref{fig:hit_rate} (right) shows how such  distance decreases as $q$ decreases, indicating that the files \greedyh{} stores tend to be cached longer and longer under \qlrudh, and partially under \qlru. The allocations of \fifo{} and \greedyh{} are instead quite far. }

Figure~\ref{fig:delay} shows the corresponding results for \greedyd{} and \qlrudd. 
The conclusion is the same: as $q$ decreases \qlrud{} improves the metric of interest (the delay in this case) achieving performance comparable to those of the optimal {offline} greedy allocation {and outperforming existing policies like~\qlru{} and \fifo.}



\begin{figure}[h]
         \centering
         \includegraphics[width=0.32\textwidth] {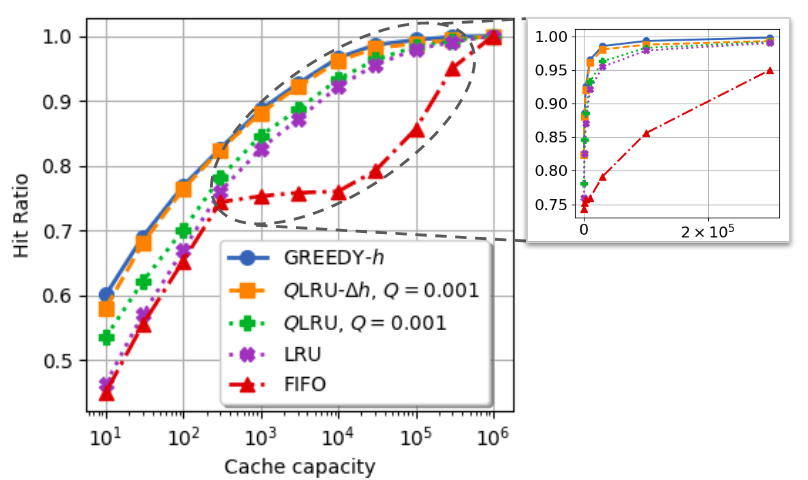}
         \caption{{Comparison of {online} policies and $\greedyh$: hit ratio versus cache capacity.}}
         \label{fig:capacity_ratio}
\end{figure}

\begin{figure}[h]
         \centering
         \includegraphics[width=0.32\textwidth] {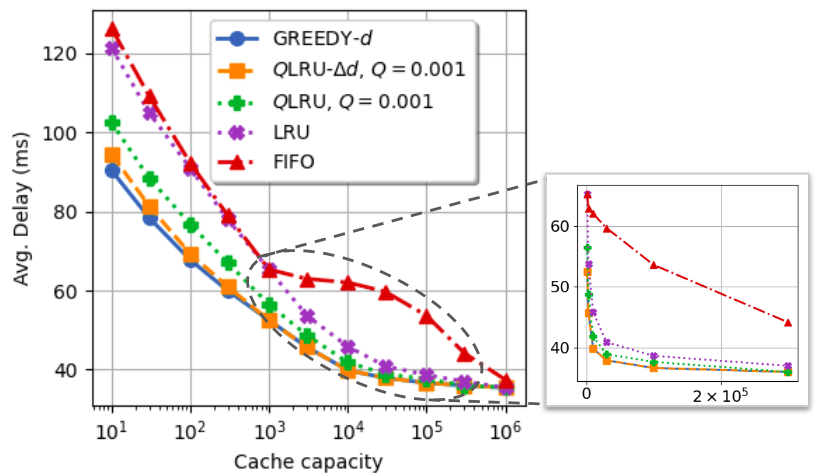}
         \caption{{Comparison of {online} policies and $\greedyd$: average delay versus cache capacity.}}
         \label{fig:capacity_delay}
\end{figure}

{Figures~\ref{fig:capacity_ratio} and~\ref{fig:capacity_delay} show the hit ratio and average delay, respectively, of {online} policies and greedy algorithms as we increase the cache capacity per BS. We fix $q=0.001$ for \qlrud{} and \qlru{}. In both scenarios, \qlrud{} outperforms all other {online} policies and it closely follows the result of the corresponding greedy policy. Note that the strange shape of  \fifo{} curves is an artefact of the semi-log graph as shown by the inserts.} 

{Additionally, we have carried out additional experiments with different catalog size and popularity distributions,
these results are  qualitatively very similar to the results already reported.}

If some knowledge about content popularity is available, it can be exploited to determine the initial content to allocate in the caches using the offline greedy algorithms, i.e., \greedyh{} and \greedyd{} when the metric of interest is the hit ratio or the delay, respectively. We show through an experiment in Fig.~\ref{fig:frame_pop} that \qlrud{} can modify the initial cache configuration and improve performance. The left figure considers the hit ratio as objective, the right one the delay.
The ground truth popularity follows a Zipf distribution with $\alpha = 1.2$ (as in the previous experiments) and noisy popularity estimations are available: they are obtained multiplying true popularities by random values from a log-normal distribution with expected value $1.0$ and variance $e^{\sigma^2} - 1$ ($\sigma^2$ is the variance of its logarithm). If $\sigma^2=0$, estimated popularity values coincide with the true ones, but  the larger the variance $\sigma^2$, the less accurate the estimations. 

The horizontal dashed lines indicate the performance of the corresponding initial cache configuration under the true request process. The solid curves show the performance over time of \qlrudh{} (left) and \qlrudd{} (right) with $q=10^{-3}$. We observe that the curves converge to the same value, that is slightly worse than the initial one when popularity estimations are exact ($\sigma^2=0$), but better in all other cases. This result shows that \qlrud{} can effectively improve performance even when  popularity estimates are available. Interestingly, one may expect that the time needed for \qlrud{} to reach the steady state performance depends on the accuracy of the initial popularity estimates (the more accurate, the less changes would be needed to reach the final cache allocation), but the dependence, if present at all, is very small.

We remark that available popularity information could also be used also to tune \qlrud's parameters to speed-up the transient. For example, we can modify~\eqref{e:miss} to favor the contents the greedy algorithm would have put in the cache. 
This change is in the same spirit of introducing the factor $\Delta g_f^{(b)}(\X_f(t) \oplus \mathbf e^{(b)},u)$ in~\eqref{e:miss}. As we discuss at the end of Section~\ref{s:operation}, these changes likely improve convergence speed, but do not affect the steady-state and then \qlrud's optimality guarantees.

\begin{figure}[t]
         \centering
         \includegraphics[width=0.5\textwidth] {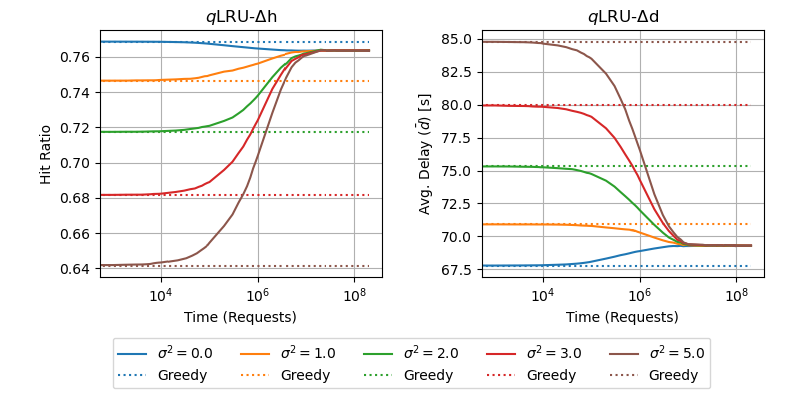}
         \caption{Convergence of \qlrudh{} (left) and \qlrudd{} (right) starting the simulation 	with the respective greedy allocation for different accuracy of popularity estimation, quantified by the variance~$\sigma^2$. The solid curves are the average of $100$ different simulation rounds.}
         \label{fig:frame_pop}
\end{figure}

\section{Discussion and conclusions}
\label{s:conclusions}
In this paper, we have introduced \qlrud, a general-purpose caching policy that can be tuned to optimize different performance metrics in a dense cellular network. Recently~\cite{garetto20}, we discovered that the same approach can be applied to a different application scenario, i.e., similarity caching systems, in which a user request for an object $o$ that is not in the cache can be (partially) satisfied by a similar stored object $o'$, at the cost of a loss of user utility. This cost can be expressed as function of the set of objects currently stored in the cache, similarly to how in this paper the gain is a function of the set of BSs storing the content.

Under stationary request processes, the smaller $q$ is, the better \qlrud{} performs.
When content popularities and/or user densities vary over time, the caching policy may react too slowly to changes if $q$ is small. A detailed experimental evaluation in~\cite{leonardi18jsac} using real traces from Akamai suggests that the sweet-spot is for $q$ values between $0.01$ and $0.1$, that achieve a good tradeoff between convergence speed and performance. A practical alternative to make the policy more reactive is to use a virtual cache. The virtual cache only stores content ids and it is managed independently from the physical cache, e.g., through a \lru{} policy. Upon a miss at a physical cache, the content is stored there if and only if its id is present in the virtual cache. Upon hits, the policy updates the state of the cache exactly as \qlrud. Under stationary request traffic, a miss for content $i$ leads to an insertion with probability $1-e^{-\lambda_i T_{c,v}}$, where $T_{c,v}$ is the characteristic time of the virtual cache. The virtual cache can be seen as an alternative way to implement a probabilistic insertion (at the cost of introducing a popularity-bias), achieving small insertion probabilities when the virtual cache (and then $T_{c,v}$) is small. At the same time, two close requests for a content cause it to be placed immediately in the physical cache, while \qlrud{} would store it on average after $1/q$ requests. This variant reacts faster and it is then more suited for non-stationary settings. 

\qlrud{} responds to hits in a binary way: the content is moved to the front or maintained in the same position. The dynamic performance of the policy may probably be improved by introducing a list-based variant~\cite{gast15}, where the cache is organized in a number of ranked lists and a content is promoted to a higher-priority list upon a hit. The marginal gain of the copy can affect the probability of the content to be randomly promoted to the next list, or the number of lists the content advances by. 

Another interesting research direction is to extend \qlrud{} to operate with heterogeneous content sizes. This can be probably achieved by making the update probability inversely proportional to the content size, similarly to what done in~\cite{neglia18ton}.  

This work was partly funded by the French Government (National Research Agency, ANR) through the ``Investments for the Future'' Program reference \#ANR-11-LABX-0031-01.

\IEEEpeerreviewmaketitle

\bibliographystyle{IEEEtran}

\bibliography{caching.bib}

\begin{thebibliography}{10}
\providecommand{\url}[1]{#1}
\csname url@samestyle\endcsname
\providecommand{\newblock}{\relax}
\providecommand{\bibinfo}[2]{#2}
\providecommand{\BIBentrySTDinterwordspacing}{\spaceskip=0pt\relax}
\providecommand{\BIBentryALTinterwordstretchfactor}{4}
\providecommand{\BIBentryALTinterwordspacing}{\spaceskip=\fontdimen2\font plus
\BIBentryALTinterwordstretchfactor\fontdimen3\font minus
  \fontdimen4\font\relax}
\providecommand{\BIBforeignlanguage}[2]{{%
\expandafter\ifx\csname l@#1\endcsname\relax
\typeout{** WARNING: IEEEtran.bst: No hyphenation pattern has been}%
\typeout{** loaded for the language `#1'. Using the pattern for}%
\typeout{** the default language instead.}%
\else
\language=\csname l@#1\endcsname
\fi
#2}}
\providecommand{\BIBdecl}{\relax}
\BIBdecl

\bibitem{CISCO}
``Cisco visual networking index: Global mobile data traffic forecast update,
  2017?2022 white paper,'' CISCO, Tech. Rep., February 2019.

\bibitem{lee12}
D.~Lee, H.~Seo, B.~Clerckx, E.~Hardouin, D.~Mazzarese, S.~Nagata, and
  K.~Sayana, ``Coordinated multipoint transmission and reception in
  {LTE}-advanced: deployment scenarios and operational challenges,''
  \emph{{IEEE} Communications Magazine}, vol.~50, no.~2, pp. 148--155, Feb.
  2012.

\bibitem{shanmugam13}
K.~Shanmugam, N.~Golrezaei, A.~Dimakis, A.~Molisch, and G.~Caire,
  ``Femtocaching: Wireless video content delivery through distributed caching
  helpers,'' \emph{IEEE Transactions on Information Theory}, vol.~59, no.~12,
  pp. 8402--8413, 2013.

\bibitem{garetto16}
M.~Garetto, E.~Leonardi, and V.~Martina, ``A unified approach to the
  performance analysis of caching systems,'' \emph{ACM Trans. Model. Perform.
  Eval. Comput. Syst.}, vol.~1, no.~3, pp. 12:1--12:28, May 2016.

\bibitem{LTE-book}
C.~Cox, \emph{An introduction to LTE: LTE, LTE-advanced, SAE and 4G mobile
  communications}.\hskip 1em plus 0.5em minus 0.4em\relax John Wiley \& Sons,
  2012.

\bibitem{Caire12}
N.~Golrezaei, K.~Shanmugam, A.~G. Dimakis, A.~F. Molisch, and G.~Caire,
  ``Femtocaching: Wireless video content delivery through distributed caching
  helpers,'' in \emph{IEEE INFOCOM 2012}, March 2012, pp. 1107--1115.

\bibitem{Poularakis14}
K.~Poularakis, G.~Iosifidis, and L.~Tassiulas, ``Approximation algorithms for
  mobile data caching in small cell networks,'' \emph{IEEE Transactions on
  Communications}, vol.~62, no.~10, pp. 3665--3677, Oct 2014.

\bibitem{saputra19}
Y.~M. {Saputra}, H.~T. {Dinh}, D.~{Nguyen}, and E.~{Dutkiewicz}, ``A novel
  mobile edge network architecture with joint caching-delivering and horizontal
  cooperation,'' \emph{IEEE Transactions on Mobile Computing}, pp. 1--1, 2019.

\bibitem{Naveen15}
K.~Naveen, L.~Massoulie, E.~Baccelli, A.~Carneiro~Viana, and D.~Towsley, ``On
  the interaction between content caching and request assignment in cellular
  cache networks,'' in \emph{5th Workshop on All Things Cellular: Oper.,
  Applic, and Challenges}.\hskip 1em plus 0.5em minus 0.4em\relax ACM, 2015.

\bibitem{Chattopadhyay18}
A.~{Chattopadhyay}, B.~{B{\l}aszczyszyn}, and H.~P. {Keeler}, ``Gibbsian
  on-line distributed content caching strategy for cellular networks,''
  \emph{IEEE Transactions on Wireless Communications}, vol.~17, no.~2, pp.
  969--981, Feb 2018.

\bibitem{Anastasios2}
B.~Blaszczyszyn and A.~Giovanidis, ``Optimal geographic caching in cellular
  networks,'' in \emph{IEEE ICC 2015}, June 2015, pp. 3358--3363.

\bibitem{avrachenkov17}
K.~Avrachenkov, J.~Goseling, and B.~Serbetci, ``A low-complexity approach to
  distributed cooperative caching with geographic constraints,'' \emph{Proc.
  ACM Meas. Anal. Comput. Syst.}, vol.~1, no.~1, pp. 27:1--27:25, Jun. 2017.

\bibitem{ao15}
W.~C. Ao and K.~Psounis, ``Distributed caching and small cell cooperation for
  fast content delivery,'' in \emph{MobiHoc}.\hskip 1em plus 0.5em minus
  0.4em\relax ACM, 2015, pp. 127--136.

\bibitem{tuholukova17}
A.~{T}uholukova, G.~{N}eglia, and T.~{S}pyropoulos, ``Optimal cache allocation
  for femto helpers with joint transmission capabilities,'' in \emph{2017 IEEE
  International Conference on Communications (ICC)}, May 2017, pp. 1--7.

\bibitem{chen17}
Z.~{Chen}, J.~{Lee}, T.~Q.~S. {Quek}, and M.~{Kountouris}, ``Cooperative
  caching and transmission design in cluster-centric small cell networks,''
  \emph{IEEE Transactions on Wireless Communications}, vol.~16, no.~5, pp.
  3401--3415, May 2017.

\bibitem{leconte16}
M.~Leconte, G.~Paschos, L.~Gkatzikis, M.~Draief, S.~Vassilaras, and
  S.~Chouvardas, ``Placing dynamic content in caches with small population,''
  in \emph{IEEE INFOCOM 2016}, 2016.

\bibitem{giovanidis16}
A.~Giovanidis and A.~Avranas, ``Spatial multi-lru caching for wireless networks
  with coverage overlaps,'' \emph{SIGMETRICS Perform. Eval. Rev.}, vol.~44,
  no.~1, pp. 403--405, Jun. 2016.

\bibitem{paschos19}
G.~S. {Paschos}, A.~{Destounis}, L.~{Vigneri}, and G.~{Iosifidis}, ``Learning
  to cache with no regrets,'' in \emph{IEEE INFOCOM 2019 - IEEE Conference on
  Computer Communications}, April 2019, pp. 235--243.

\bibitem{wu2019dynamic}
P.~Wu, J.~Li, L.~Shi, M.~Ding, K.~Cai, and F.~Yang, ``Dynamic content update
  for wireless edge caching via deep reinforcement learning,'' \emph{IEEE
  Communications Letters}, vol.~23, no.~10, pp. 1773--1777, 2019.

\bibitem{leonardi18jsac}
E.~Leonardi and G.~Neglia, ``Implicit coordination of caches in small cell
  networks under unknown popularity profiles,'' \emph{IEEE Journal on Selected
  Areas in Communications}, vol.~36, no.~6, pp. 1276--1285, June 2018.

\bibitem{che02}
H.~Che, Y.~Tung, and Z.~Wang, ``{Hierarchical Web caching systems: modeling,
  design and experimental results},'' \emph{Selected Areas in Communications,
  IEEE Journal on}, vol.~20, no.~7, pp. 1305--1314, Sep 2002.

\bibitem{connors88}
D.~P. Connors and P.~R. Kumar, ``Balance of recurrence order in
  time-inhomogenous markov chains with application to simulated annealing,''
  \emph{Probability in the Engineering and Informational Sciences}, vol.~2,
  no.~2, pp. 157--184, 1988.

\bibitem{connors89}
------, ``Simulated annealing type markov chains and their order balance
  equations,'' \emph{SIAM Journal on Control and Optimization}, vol.~27, no.~6,
  pp. 1440--1461, 1989.

\bibitem{desai94}
M.~Desai, S.~Kumar, and P.~R. Kumar, ``{Quasi-Statically Cooled Markov
  Chains},'' \emph{Probability in the Engineering and Informational Sciences},
  vol.~8, no.~1, pp. 1--19, 1994.

\bibitem{arxiv1}
\BIBentryALTinterwordspacing
G.~Neglia, E.~Leonardi, G.~Iecker, and T.~Spyropoulos, ``A swiss army knife for
  dynamic caching in small cell networks,'' \emph{CoRR}, vol. abs/1912.10149,
  2019. [Online]. Available: \url{http://arxiv.org/abs/1912.10149}
\BIBentrySTDinterwordspacing

\bibitem{young93}
H.~P. Young, ``{The Evolution of Conventions},'' \emph{Econometrica}, vol.~61,
  no.~1, pp. 57--84, January 1993.

\bibitem{kelly14stochastic_networks}
F.~Kelly and E.~Yudovina, \emph{Stochastic {Networks}}, ser. Institute of
  {Mathematical} {Statistics} {Textbooks}.\hskip 1em plus 0.5em minus
  0.4em\relax Cambridge: Cambridge University Press, 2014.

\bibitem{fagin77}
R.~Fagin, ``Asymptotic miss ratios over independent references,'' \emph{Journal
  of Computer and System Sciences}, vol.~14, no.~2, pp. 222 -- 250, 1977.

\bibitem{garetto15}
M.~Garetto, E.~Leonardi, and S.~Traverso, ``Efficient analysis of caching
  strategies under dynamic content popularity,'' in \emph{IEEE INFOCOM 2015},
  April 2015, pp. 2263--2271.

\bibitem{Jele99}
P.~R. Jelenkovic, ``Asymptotic approximation of the move-to-front search cost
  distribution and least-recently used caching fault probabilities,'' \emph{The
  Annals of Applied Probability}, vol.~9, no.~2, pp. 430--464, 1999.

\bibitem{fricker2012}
C.~Fricker, P.~Robert, and J.~Roberts, ``{A versatile and accurate
  approximation for LRU cache performance},'' in \emph{Proceedings of the 24th
  International Teletraffic Congress}, 2012, p.~8.

\bibitem{anantharam89}
V.~Anantharam and P.~Tsoucas, ``A proof of the markov chain tree theorem,''
  \emph{Statistics \& Probability Letters}, vol.~8, no.~2, pp. 189 -- 192,
  1989.

\bibitem{tse2005fundamentals}
D.~Tse and P.~Viswanath, \emph{Fundamentals of wireless communication}.\hskip
  1em plus 0.5em minus 0.4em\relax Cambridge university press, 2005.

\bibitem{bs_dataset}
\BIBentryALTinterwordspacing
``Openmobilenetwork.'' [Online]. Available:
  \url{http://map.openmobilenetwork.org/}
\BIBentrySTDinterwordspacing

\bibitem{garetto20}
M.~Garetto, E.~Leonardi, and G.~Neglia, ``{Similarity Caching: Theory and
  Algorithms},'' in \emph{{Proceedings of the IEEE Conference on Computer
  Communications (Infocom) 2020}}, Beijing, China, Apr. 2020.

\bibitem{gast15}
N.~Gast and B.~Van~Houdt, ``Transient and steady-state regime of a family of
  list-based cache replacement algorithms,'' \emph{SIGMETRICS Perform. Eval.
  Rev.}, vol.~43, no.~1, pp. 123--136, Jun. 2015.

\bibitem{neglia18ton}
G.~Neglia, D.~Carra, and P.~Michiardi, ``{Cache Policies for Linear Utility
  Maximization},'' \emph{IEEE/ACM Transactions on Networking}, vol.~26, no.~1,
  pp. 302--313, 2018.

\bibitem{bertsekas99nonlinear}
D.~P. Bertsekas, \emph{\BIBforeignlanguage{English}{Nonlinear {Programming}}},
  2nd~ed.\hskip 1em plus 0.5em minus 0.4em\relax Belmont, Mass: Athena
  Scientific, Sep. 1999.

\end{thebibliography}

\appendices

\section{Proof of Lemma~\ref{l:our_balance}}
\label{a:our_balance}
 \begin{proof}
Consider the function:
\begin{equation}
\label{e:potential2}
\phi_f(\x_{f})\triangleq G_f(\x_f) - \vgamma^\intercal \x_f.
\end{equation}
We show that $\{\phi_f(\x_{f})\}$ is a solution of the system~\eqref{e:balance_eqs} (for a particular value of $\sigma$). 

To this purpose, for a given choice of the set $A$, we need to evaluate $\underset{\x_f \in A, \z_f \in A^c}{\max}   \phi_f(\x_f)    -   r_f(\x_f,\z_f) $. 
We start proving that the maximum is always achieved by a pair of parent-child nodes. In particular, we show that for any two states $\hat{\x}_f \in A$ and $\hat{\z}_f \in A^c$  with $r_f(\hat{\x}_f,\hat{\z}_f) < \infty$ and $|\hat{\z}_f| > |\hat{\x}_f|+1$, (which imply that $\hat{\z}_f $ is an ancestor of $\hat{\x}_f$), there exist two states $\x'_f \in A$ and $\y'_f \in A^c$, with~$\y'_f$ parent of $\x'_f$ and  
\begin{equation}
\label{e:higher_transitions}
\phi_f(\hat\x_f) - r_f(\hat \x_f,\hat \z_f) \le \phi_f(\x'_f) - r_f(\x'_f,\y'_f).
\end{equation}
Consider a path from $\hat{\z}_f$ to $\hat{\x}_f$ that traverses states with strictly smaller weight (it is obtained setting progressively to zero the elements that are equal to one in $\hat{\z}_f$, but not in $\hat{\x}_f$). One of the edges of this path necessarily goes from a state in $A^c$ to a state in $A$. These two states are respectively $\y_f'$ and $\x'_f$. In fact, 
\begin{align*}
&\phi_f(\hat{\x}_f) - r_f(\hat{\x}_f,\hat{\z}_f) = \\
&\quad = G_f(\hat{\x}_f) - \vgamma^\intercal \hat{\x}_f - \vgamma^\intercal (\hat{\z}_f - \hat{\x}_f)\\
&\quad	= G_f(\hat{\x}_f) - \vgamma^\intercal \hat{\z}_f\\
&\quad	\le G_f(\x'_f) - \vgamma^\intercal \hat{\z}_f\\
&\quad	=  G_f(\x'_f) - \vgamma^\intercal \x'_f  - \vgamma^\intercal (\y'_f - \x'_f) - \vgamma^\intercal (\hat{\z}_f - \y'_f)\\
&\quad	=  \phi_f(\x'_f) - r_f(\x'_f, \y'_f)  - \vgamma^\intercal (\hat{\z}_f - \y'_f)\\
&\quad  \le \phi_f(\x'_f) - r_f(\x'_f, \y'_f), 
\end{align*}
where the first inequality follows from the monotonicity of $G_f(\cdot)$, and the second from the fact that $\z_f$ is an ancestor of $\y'_f$ and then $\hat{\z}_f - \y'_f$ is a vector with non-negative elements. 

In addition note that by construction
$r(\hat{\z}_f,\hat{\x}_f)=\infty$ (i.e. given two states $\z_f$ and $\x_f$ with $|\z_f|>|\x_f|$, we have $r(\z_f,\x_f)<\infty$ only if  $\x_f$ is a child of $\z_f$).
 
As a consequence we have  that $\{\phi(\x_f)\}$ is a solution of system \eqref{e:balance_eqs}, if and only if it is a solution of 
\begin{equation}
\label{e:balance_eqs2}
\left\{
\begin{aligned}
&  \underset{\begin{subarray}{c}\x_f \in A, \z_f \in A^c,\\ |\z_f| = |\x_f| \pm 1\end{subarray}}{\max}   \nu_f(\x_f)    -   r_f(\x_f,\z_f)  =  \\
& \quad = \underset{\begin{subarray}{c}\x_f \in A, \z_f \in A^c,\\ |\z_f| = |\x_f| \pm 1\end{subarray}}{\max}    \nu_f(\z_f) - r_f(\z_f,\x_f), \forall A \subset \{0,1\}^B\\
&  \underset{\x_f \in \{0,1\}^B}{\max}   \nu_f(\x_f)  = \sigma.
\end{aligned}
\right.
\end{equation}
We can then limit ourselves to check if $\phi_f(\cdot)$ satisfies the aggregate balance equations considering only the parent-child pairs.
We prove a stronger relation, i.e., that for every parent-child pair, $\phi_f(\cdot)$ satisfies a pairwise balance equation. In fact, for every $\x_f$ and $\y_f$ with $\y_f\!= \!\x_f\!\oplus\!\mathbf e^{(b_0)}$ and  parent of $\x_f$:
\begin{align*}
\phi_f(\x_f) - r_f(\x_f,\y_f) & = G_f(\x_f) - \vgamma^\intercal \x_f - \vgamma^\intercal (\y_f - \x_f)\\
	& = G_f(\x_f) - \vgamma^\intercal \y_f \\
	& =  G_f(\y_f) - \Delta G_f^{(b_0)}(\y_f) - \vgamma^\intercal \y_f\\
	& =  \phi_f(\y_f) - r_f(\y_f, \x_f).
\end{align*}
It follows that $\{\phi(\x_f)\}$ is a solution of~system~\eqref{e:balance_eqs}.
\end{proof}

\section{Proof of Proposition~\ref{p:qlrud_convergence_general}}
\label{a:proof_qlrud}
\begin{proof}
From Corollary~\ref{c:stochastically_stable} a state $\x_f$ is stochastically stable if and only if it is a global maximizer of $\phi_f(\cdot)$, i.e., $\lim_{q\to 0 } \pi_{f,q}(\x_f) >0 $ if and only if $\x_f$ is a maximizer of~$\phi_f(\cdot)$.

Let $\pi_{f,0^+} (\x_f) \triangleq \lim_{q\to 0 } \pi_{f,q}(\x_f) $  denote the limit of the probability distribution. We are now going to prove that the $\{\pi_{f,0^+}(\x_f), f \in [F], \x_f \in {0,1}^B \}$ is an optimal solution for problem~\eqref{e:relaxed_opt}. 

Problem~\eqref{e:relaxed_opt} is a convex problem. We can consider its Lagrangian function:
\begin{equation}
\begin{aligned}
&L(\boldsymbol \alpha,  \boldsymbol \chi, \boldsymbol \zeta) = - \sum_{f=1}^F \sum_{\x_f \in \{0,1\}^B} \alpha_f(\x_f) G_f(\x_f) \\
&\quad   + \sum_{b=1}^B \chi_b \left( \sum_{f=1}^F \sum_{\x_f \in \{0,1\}^B} \alpha_f(\x_f) x_f^{(b)} - C \right) \\
&\quad + \sum_{f=1}^F \zeta_f \left( \sum_{\x_f \in \{0,1\}^B} \alpha_f(\x_f) - 1 \right),
\end{aligned}
\end{equation}
where $\boldsymbol \alpha $ denotes the $F 2^B$ vector of problem variables, $\boldsymbol \chi$ denotes the $B$ vector of Lagrange multipliers relative to the capacity constraints, and $\boldsymbol \zeta$ denotes the $F$ vector of Lagrange multipliers relative to the total mass to assign to each file.

A vector $\boldsymbol \alpha^* $ is a (global) maximizer of problem~\eqref{e:relaxed_opt}, if there are vectors $\boldsymbol \chi^*$ and $\boldsymbol \zeta^*$ such that~\cite[Thm.~3.4.1]{bertsekas99nonlinear}
\begin{align*}
1)\;\; &  \boldsymbol \alpha^*  \textrm{ is feasible,}\\
2)\;\;  & \nabla  L(\boldsymbol \alpha^*, \boldsymbol \chi^*, \boldsymbol \zeta^*)^\intercal \left(\boldsymbol \alpha - \boldsymbol \alpha^*  \right)\ge 0, \;\; \forall \boldsymbol \alpha \ge \boldsymbol 0.
\end{align*}

We show that the following assignments satisfy the set of conditions indicated above
\begin{align*}
	 \alpha^*_f(\x_f) &  =  \pi_{f,0^+}(\x_f), & &  \forall f \in [F], \x_f \in \{0,1\}^B,\\
	\chi^*_b   & =    \gamma_b, & &   \forall b \in [B],\\
	\zeta_f^*  & = \max_{\x'_f \in \{0,1\}^B} \phi_f(\x'_f), & & \forall f \in [F].
\end{align*}
In fact, for any value $q$,  $\sum_f \sum_{\x_f} x_f^{(b)} \pi_{f,q}(\x_f)= C $ for each $b$, $\sum_{\x_f}\!\pi_{f,q}(\x_f)\!=\!1 $ for each $f$, and $\pi_{f,q}(\x_f)\!\ge\!0$ for each $f$ and $\x_f$. The same relations are also satisfied passing to the limit when $q$ converges to $0$, then~$\{\pi_{f,0^+}(\x_f)\}$ is a feasible solution. Finally,
\begin{align*}
&\frac{\partial L(\boldsymbol \alpha, \boldsymbol \chi,  \boldsymbol \zeta)}{\partial \alpha_f(\x_f)}  \biggr|_{\substack{\boldsymbol \alpha =\boldsymbol \alpha^*  \boldsymbol \chi = \boldsymbol \chi^* \\ \boldsymbol \zeta = \boldsymbol \zeta^*}} =\\
&\quad = -G_f(\x_f) + \sum_{b=1}^B \gamma_b x_f^{(b)} + \max_{\x'_f \in \{0,1\}^B} \phi_f(\x'_f) \\
&\quad =  -\phi(\x_f) + \max_{\x'_f \in \{0,1\}^B} \phi_f(\x'_f)\\
	&\quad  \begin{cases}
		= 0 & \textrm{ if $\x_f$ is stochastically stable,}\\
		> 0 & \textrm{ otherwise.}
	\end{cases}
\end{align*}
Let $\mathcal S_f \subset \{0,1\}^B$ denote the set of stochastically stable states for file $f$. It follows that
\begin{align*}
&\nabla  L (\boldsymbol \alpha^*,  \boldsymbol \chi^*, \boldsymbol \zeta^*)^\intercal   \left(\boldsymbol \alpha - \boldsymbol \alpha^*  \right) = \\
& = \sum_{f=1}^{F} \sum_{\x_f \in \{0,1\}^B} \frac{\partial L(\boldsymbol \alpha, \boldsymbol \chi,  \boldsymbol \zeta)}{\partial \alpha_f(\x_f)}  \biggr|_{\substack{\boldsymbol \alpha =\boldsymbol \alpha^* \\ \boldsymbol \chi = \boldsymbol \chi^* \\ \boldsymbol \zeta = \boldsymbol \zeta^*}}\times (\alpha_f(\x_f) - \alpha^*_f(\x_f) )\\
& = \sum_{f=1}^{F} \sum_{\x_f \in \mathcal S_f} 0 \times (\alpha_f(\x_f) - \pi_{f,0^+}(\x_f) ) \\
& +  \sum_{f=1}^{F} \sum_{\x_f \notin \mathcal S_f} \frac{\partial L(\boldsymbol \alpha, \boldsymbol \chi,  \boldsymbol \zeta)}{\partial \alpha_f(\x_f)}  \biggr|_{\substack{\boldsymbol \alpha =\boldsymbol \alpha^* \\ \boldsymbol \chi = \boldsymbol \chi^* \\ \boldsymbol \zeta = \boldsymbol \zeta^*}} \times (\alpha_f(\x_f) - 0 )\\
		& \ge 0.
\end{align*}
\end{proof}

\end{document}